\documentclass[a4paper,12pt]{article}
\usepackage{latexsym,amssymb,amsfonts,amsmath,amsthm}
\usepackage[dvips]{graphicx}
\usepackage{comment}
\usepackage{color}

\numberwithin{equation}{section}

\newtheorem{theorem}{Theorem}
\newtheorem{lemma}{Lemma}
\newtheorem{corollary}{Corollary}

\newtheorem{remark}{Remark}
\newtheorem{definition}{Definition}

\newcommand{\bs}[1]{\boldsymbol{#1}}

\newcommand{\sgn}{{\rm sgn}}

\newcommand{\ket}[1]{|{#1}\rangle}
\newcommand{\bra}[1]{\langle{#1}|}

\title{{\Large {\bf Sensitivity of quantum walks to a boundary of
two-dimensional lattices: approaches based on the CGMV method and topological
phases
}
}}
\author{ 
{\small 
Takako Endo,$^{1}$ 
Norio Konno,$^{2}$ 
Hideaki Obuse,$^{3}$ 
Etsuo Segawa. $^{4}$ 
\footnote{
e-segawa@m.tohoku.ac.jp 
}
}\\ 
{\scriptsize $^1$ 
Institute for Global Leadership, Ochanomizu University
}\\
{\scriptsize
2-1-1 Ohtsuka, Bunkyo, Tokyo, 112-0012, Japan 
}\\
{\scriptsize $^{2}$ 
Department of Applied Mathematics, Faculty of Engineering, Yokohama National University
}\\
{\scriptsize 
Hodogaya, Yokohama 240-8501, Japan
} \\
{\scriptsize $^3$ 
Department of Applied Physics, Hokkaido University
}\\
{\scriptsize 
Sapporo, Hokkaido 060-8628, Japan
} \\
{\scriptsize $^4$ 
Graduate School of Information Sciences, Tohoku University, 
}\\
{\scriptsize 
Aoba, Sendai 980-8579, Japan
} \\
} 
\date{\empty }
\begin{document}
\maketitle

\par\noindent
\begin{small}
\par\noindent
{\bf Abstract}. 
In this paper, we treat quantum walks in a two-dimensional lattice with
cutting edges along a straight boundary introduced by Asboth and Edge [Phys.\ Rev.\ A
 {\bf 91}, 022324 (2015)] in
 order to study one-dimensional edge states originating from topological
 phases of matter and 
 to obtain collateral evidence of how a quantum walker reacts
 to the boundary.
Firstly, we connect this model to the CMV matrix, which provides a $5$-term recursion relation of the Laurent polynomial 
associated with spectral measure on the unit circle. 
Secondly, we explicitly derive the spectra of bulk and edge states of the quantum
 walk with the boundary using spectral analysis of the CMV matrix. 
Thirdly, while topological numbers of the model studied so far are
 well-defined only when gaps in the bulk spectrum exist, we find a new topological
number defined only when there are no gaps in the bulk spectrum.
We confirm that the existence of the spectrum for edge states derived
 from the CMV matrix is consistent
 with the prediction from a bulk-edge correspondence using topological numbers
 calculated in the cases where gaps in the bulk spectrum do or do
 not exist.
Finally, we show how the edge states contribute to the asymptotic behavior of the quantum walk through limit theorems of the finding probability.
Conversely, we also propose a differential equation using this limit distribution whose solution is the underlying edge state.  

\end{small}

\setcounter{equation}{0}

\section{Introduction}
Quantum walks are the quantum analog of random walks, first introduced
in \cite{Gudder}.  The initial impetus for intensively studying quantum walks came from the
field of  quantum information~\cite{ABNVW}. 
Nowadays, overlaps between quantum walks and quantum information, as
well as, various other research fields have been found and
established by interdisciplinary researches (see \cite{KW, Portugal} and references therein). 

One recent area of studying is
connecting localization in quantum walks with topological
phases of matter, i.e., topological insulators, which is currently
a hot field in condensed-matter physics.
Kitagawa {\it et al.} \cite{Kitagawa,Kitagawa2} introduced a quantum
walk as a model for exploring the topological phases of matter,
which is simplified by discretizing temporal and spatial spaces. 
In particular, quantum walks on a one-dimensional lattice, 
whose unitary time-evolution operator is spatially inhomogeneous
and retains chiral symmetry, 
have been intensively studied (e.g., \cite{A,CGGSVWW,CGSVWW,EKO,EEKST,OK,OANK}). 
Kitagawa\cite{Kitagawa2} also found a connection between 
two-dimensional topological insulators and quantum walk models on a
two-dimensional square lattice $(x,y) \in \mathbb{Z}^2$~\cite{FMB,FMMB}. 
Asboth and Edge further studied on this model with  a ``cutting
edge'' ~\cite{AE}, meaning that all the edges
connecting vertices
at $(0,y)\in \mathbb{Z}^2$ and $(-1,y)\in \mathbb{Z}^2$ for $y\in
\mathbb{Z}$ 
are cut and rewired to make
self-loops on $(0,y)$ and $(-1,y)$. 
This model is determined by a pair of parameters $(\alpha,\beta)\in [0,2\pi)^2$; 
the first and second parameters, $\alpha$ and $\beta$, determine the local dynamics of the horizontal and vertical directions on the two-dimensional lattice, 
respectively. The local dynamics are provided by alternatively
allowing the two-dimensional rotation matrices with angles of $\alpha$ and
$\beta$ to act up on two-dimensional internal states;    
we call this the AE model. 
See for a detailed definition of this graph setting and the time evolution in Section~2. 
Asboth and Edge provided analytical and numerical results showing the
existence of unidirectional edge states along the self-loops. 

Due to the unitarity of the time evolution of quantum walks, we
obtain a sequence of probability distributions for each time step $\{\rho_n\}_{n\in\mathbb{N}}$.
Let $\psi_n$ be a quantum state at time $n$ that is obtained by the $n$-th
iteration of the unitary time evolution to an initial state. 
We call the map $\psi_n\mapsto \rho_n$ a ``measurement" which is represented by an orthogonal projection map, see Section 2.2 (3) for more detail.
In this paper, we provide further analytical results on the AE model connecting this measurement. 
The main purpose of this paper is to determine whether it is possible to
estimate spectral information about the AE model by obtaining the
distribution, $\rho_n$, after measuring the quantum state $\psi_n$, and
if so, which of the spectral properties are
reflected to the limit behavior of $\rho_n$. 

To this end, we first connect the AE model to the CMV matrix (Theorem~1). 
The CMV matrix represents the five-term recursion relation for
the orthogonal Laurent polynomials associated with a given positive measure on the 
unit circle in the complex plane~\cite{CMV1}. 
The authors Cantero, Gr{\"u}nbaum, Moral, and Vel{\'a}zquez first connected the CMV matrix to quantum walks in \cite{CGMV}. 
We call spectral analysis of quantum walks using this connection to the
CMV matrix the CGMV method after these authors' initials~\cite{KS}. 
This method enable us to obtain spectral information about our model in quite-explicit form. 
Due to the translation symmetry with respect to the direction
parallel to the boundary of the cutting edges of the AE model, we can take the Fourier transform and decompose the time operator $U^2$ 
restricted to the subspace generated by horizontal arcs  
into the unitary operators $\{\hat{\Gamma}_k\}_{k\in[0,2\pi)}$. 
Interestingly, we find that $\hat{\Gamma}_k$ for each $k$ is unitarily
equivalent to the CMV matrix with null odd Verblunsky parameters 
	\[ (\eta(k),0,\eta(k),0,\dots),\;\;\eta(k)=\sin(\alpha-\beta)\cos k+i\sin(\alpha+\beta)\sin k. \] 
Thanks to this connection, secondly, we can derive the spectra 
for the bulk and edge states.
We obtain the bulk spectrum, 
in which gaps may exist and the spectrum for the edge state, 
which may appear in the gaps of the bulk spectrum (Theorem~2).
As a corollary, we can classify the spectra according to
the signs given by 
	\[ (\sgn(\sin 2\alpha\sin
	2\beta),\sgn(\sin(\alpha+\beta)),\sgn(\sin(\alpha-\beta))). \]

Thirdly, in contrast to the topological number in Ref.\ \cite{AE} which is defined
only when gaps in the bulk spectrum exist,
we derive a new topological number which is defined only when there
are no gaps in the bulk spectrum.
We determine the relationship between the edge spectrum derived from
Theorem~2 and topological numbers calculated in the cases where gaps
in the bulk spectrum do or do not exist in figures 2--5. 

Fourthly, we examine how the boundary affects the quantum walk in the two-dimensional case (Theorem~3). 
On the one-dimensional lattice, if the edge state exists, then we
observe localization (i.e., 
the time-averaged limit measure is strictly positive) 
as the asymptotic behavior of the sequence of the probability distributions $\{\rho_n\}_n$~\cite{EKO,OK}. 
To express the response of the quantum walk to the boundary in the two-dimensional case, we define $\nu_n: \partial V \to [0,1]$ such that 
$\nu_n(j)$ is the finding probability of the AE model at the $n$-th iteration and at self-loops on 
$[0,j]\in \mathbb{Z}_+\times \mathbb{Z}$ starting from the self-loop at the origin. 
We show that the response is classified into four behaviors depending on the values of the parameters $(\alpha,\beta)$: 
\begin{enumerate}
\item\label{conti} continuous linear spreading;
\item\label{ball} ballistic spreading; 
\item\label{loc} localization;
\item  null.
\end{enumerate} 
In particular, the limit-density function of (\ref{conti}) is described using the Konno distribution~\cite{Konno,Konno2}. 
The limit behavior of the quantum walk contributed by the edge state in the case of (2) is completely expressed in Theorem~4. 

As a consequence, by assuming the group velocity $v(k)$, quasi-effective mass $M(k)$, and  density of the edge state $m_0(k)$, 
we obtain a parametric plot of the limit-density function of (2) using these physical quantities 
(see Corollary~\ref{parameterexpression}.)
By this parametric plot expression, conversely, we can estimate the underlying edge state from the obtained distribution. 
More precisely, if we have an estimate $\tilde{g}$ of the true limit distribution $g$ 
from the distribution data obtained by measurement of self-loops after finite unitary time iterations, 
then the underlying edge state can be also estimated from this data by 
solving a differential equation with some constant $c\in \mathbb{R}$; 
$\tilde{g}(v)\;dv/dk = c v^2$, implying that the group velocity of
edge state is an inverse of Konno distribution in this model (see
Corollary~4). 
We note that comparing discrete-time quantum walks with the discrete-Schr{\"o}dinger operator on
lattices~\cite{ASV,GP,SV}
is one of interesting problems, which should be addressed in future.

This paper is organized as follows. 
In Sect.~2, we define the graph and the quantum walk treated here. 
In Sect.~3, we connect the quantum walk with the CMV matrix.
Section~4 is devoted to the spectral analysis and topological phases of the AE model.
In Sect.~5, we present how the properties of the spectrum obtained in the
previous section are reflected in the stochastic behavior along the
boundary by studying limit distribution functions.
Finally, in Sect.~6, we present limit distribution functions toward the
bulk, which clearly exhibit exponential decay.
\section{Model}
\begin{figure}[ht]
\hspace{2.3cm}(a)\hspace{5.3cm}(b)
  \begin{center}
    \begin{tabular}{c}
      \begin{minipage}{0.4\hsize}
        \begin{center}
          \includegraphics[clip, width=6cm]{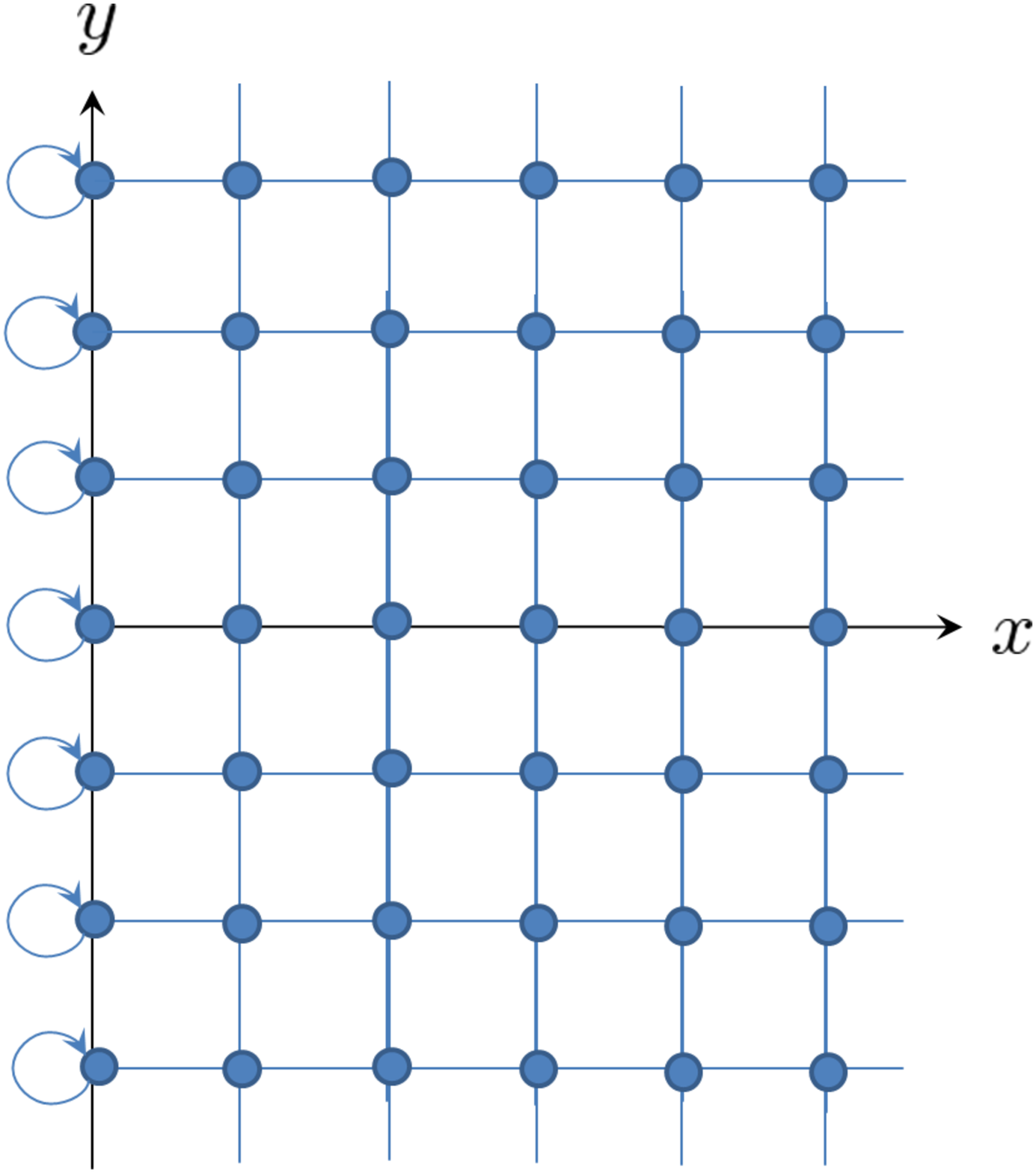}
          \hspace{1.6cm} 
        \end{center}
      \end{minipage}

      \begin{minipage}{0.4\hsize}
        \begin{center}
          \includegraphics[clip, width=7.50cm]{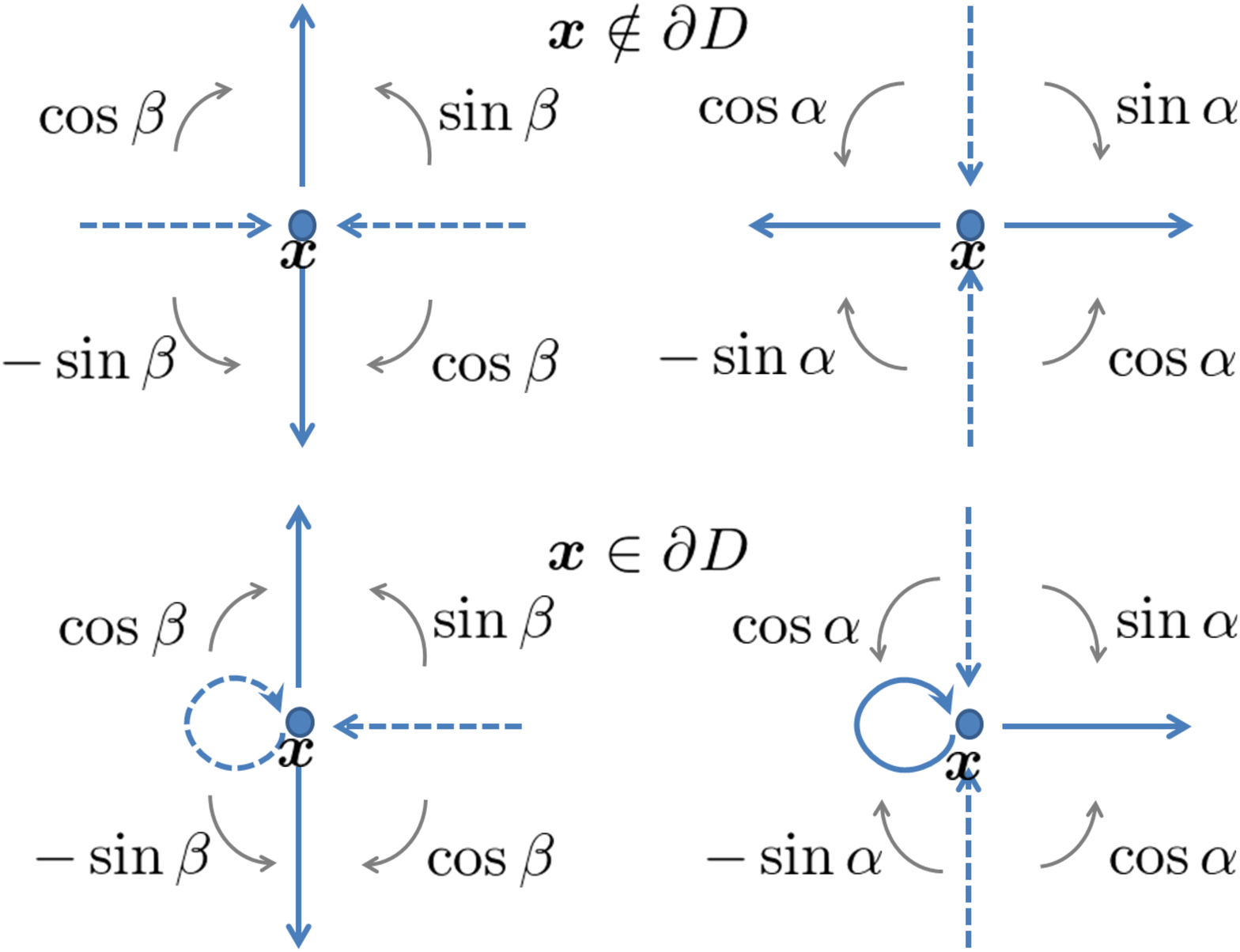}
          \hspace{1.6cm} 
        \end{center}
      \end{minipage}
    \end{tabular}
    \caption{(a) The graph treated in the present work. (b) The one-step
   time evolution $U$ at a vertex $\bs{x} \notin \partial V$ and $\bs{x} \in \partial V$. 
    The dotted  and solid arcs represent the input and output arcs, respectively.}
    \label{fig:lena}
  \end{center}
\end{figure}
\subsection{Graph}
Let $G=(V,A)$ be a directed graph whose vertex set and arc set are $V$ and $A$, respectively. 
In this paper, we express the vertex set as 
	\[ V=\mathbb{Z}_+ \times \mathbb{Z}=\{ (x,y)\;|\; x\in \mathbb{Z}_+, y\in \mathbb{Z} \}, \]
where $\mathbb{Z}$ is the integer set and $\mathbb{Z}_+$ is the non-negative integer set. 
We denote the boundary of this graph by $\partial V := \{ (0,j)\in V \;|\; j\in \mathbb{Z} \}$. 
The arc set $A$ is represented by 
	\[ A=\{ (\bs{x};d)\;|\; \bs{x}\in V, d\in\{0,1,2,3\}  \}. \]
For $a\in A$, the origin and terminus of $a$ are denoted by $o(a), t(a)\in V$, respectively, and the inverse of $a$ is denoted by $\bar{a}$. 
Here if $a=(\bs{x};d)\in A$, then
	\begin{align*} 
         t(a) &= \bs{x}, \\
         o(a) &= 
         \begin{cases}
         \bs{x}+(1,0) & \text{: $d=0$,}\\
         \bs{x} & \text{: $d=1$, $\bs{x}\in \partial V$,}\\
         \bs{x}-(1,0) & \text{: $d=1$, $\bs{x}\notin \partial V$,}\\
         \bs{x}+(0,1) & \text{: $d=2$,}\\
         \bs{x}-(0,1) & \text{: $d=3$,}
         \end{cases}
        \end{align*}
which means that $(\bs{x};d)$ is the arc coming from the direction labeled by $d$ whose terminus is $\bs{x}$; 
see figure 1. 
We call $\{ (\bs{x};1)\;|\; \bs{x}\in \partial V \}\subset A$ a set of self-loops. We take $\bar{a}=a$ for every self-loop. 
\subsection{Definition of the AE model}\label{2.2}
The one-step time evolution of the AE model is explained as follows: 
the incident state coming from the horizontal (vertical) direction is
transmitted at the terminal vertex to the vertical (horizontal)
directions with some complex-valued weight conserving the unitarity. 
Here, we define the model in the arc-set representation $\ell^2(A)$ first and then, 
convert the representation to the vertex-set ones; $\ell^2(V; \mathbb{C}^4)$ or $\ell^2(V; \mathbb{C}^2)$.
\begin{enumerate}
\item Total Hilbert space: 
$\mathcal{A}:=\ell^2(A)=\{ \psi: A\to \mathbb{C}\;|\; ||\psi||<\infty \}. $
The inner product is the standard inner product. 
Now we introduce a binary relation of $A$: 
	\[ a \stackrel{\pi}{\sim} b \Leftrightarrow t(a)=t(b) \] 
Since $\pi$ is an equivalence relation, we obtain the quotient sets by 
	\[ A/\pi = \oplus_{\bs{x} \in V} A_{\bs{x}}, \]
where $A_{\bs{x}}=\{a\in A \;|\; t(a)=x\}$.
From this partition of $A$, we set the orthogonal decomposition of $\mathcal{A}$ into 
	\[ \mathcal{A}=\bigoplus_{\bs{x}\in V} \mathcal{A}_{\bs{x}}. \]
Here 
	\[ \mathcal{A}_{\bs{x}}=\{\psi\in \mathcal{A}\;|\; a\notin A_{\bs{x}} \Rightarrow \psi(a)=0\}. \]
\item Time evolution:
We assign the local unitary operator on $\mathcal{A}_{\bs{x}}$ for each $\bs{x}\in V$; 
this operator acts alternatively in the horizontal and vertical directions in the following way 
\begin{align*} 
        C_{\bs{x}}\delta_{(\bs{x};0)} &= \cos \beta \; \delta_{(\bs{x};3)}+\sin \beta\; \delta_{(\bs{x};2)}, \\
        C_{\bs{x}}\delta_{(\bs{x};1)} &= -\sin \beta \; \delta_{(\bs{x};3)}+\cos \beta \; \delta_{(\bs{x};2)}, \\
        C_{\bs{x}}\delta_{(\bs{x};2)} &= \cos \alpha \; \delta_{(\bs{x};1)}+\sin \alpha \; \delta_{(\bs{x};0)}, \\
        C_{\bs{x}}\delta_{(\bs{x};3)} &= -\sin \alpha \; \delta_{(\bs{x};1)}+\cos \alpha \; \delta_{((\bs{x};0)},
        \end{align*}
where $\delta_a\in \mathcal{A}$ is the Kronecker delta function, that is, 
	\[ \delta_a(a')=\begin{cases} 1 & \text{: $a=a'$,}\\ 0 & \text{: $a\neq a'$.} \end{cases} \]
The pair of parameters, $(\alpha, \beta)\in [0,2\pi)^2$, which determines this quantum walk, 
will be important for providing its behavior. 
The time evolution is denoted by the iteration of the unitary map, $U: \mathcal{A}\to \mathcal{A}$ 
	\begin{equation}
        (U\psi)(a)=\sum_{b:t(b)=o(a)}\langle \delta_{\bar{a}},C_{o(a)}\delta_b \rangle \psi(b). 
        \end{equation}
See also figure 1. 
\item Distribution: 
Throughout this paper, we fix the initial state $\psi_0\in \mathcal{A}$ by $\psi_0(a)=\delta_{(\bs{0};1)}(a)$ 
which represents the self-loop at the origin. 
We define the distribution at each time step $n$, $\rho_n:A\to [0,1]$
\[ \rho_n(a)=|\psi_n(a)|^2=||\Pi_{\delta_{a}}\psi_n||^2, \]
where for a subset $\mathcal{A}'\subset \mathcal{A}$, $\Pi_{\mathcal{A}'}$ is the projection operator onto $\mathcal{A}'$. 
Here $\psi_n=U^{2n} \psi_0$. 
Obviously $\sum_{a\in A}\Pi_{\delta_{a}}=1$ holds. We regard $\rho_n(a)$ as the finding probability of the quantum walk 
at position $a\in A$ and at time $2n$. 
We observe how this quantum walk recognizes the boundary of the graph 
through the asymptotic behavior of $\nu_n(j):=\rho_{n}((0,j);1)$ for large $n$.   
\end{enumerate}
\begin{remark}
The unitarity of the time evolution $U$ is immediately shown by the
 product of the two unitary operators of $U$; $U=SC$, where
\begin{align}\label{S}
	S\psi(a)=\psi(\bar{a}),\;\;  C=\oplus \sum_{\bs{x}\in V}C_{\bs{x}}.
\end{align}
The unitary operators $S$ and $C$ are called shift and coin operators, respectively.  
\end{remark}

\subsection{An alternative quantum walk}
The horizontal and vertical arcs are denoted by 
	\[ A^{(\leftrightarrow)}=\{(\bs{x};d)\in A \;|\; d\in \{0,1\} \},
        \;\;A^{(\updownarrow)}=\{(\bs{x};d)\in A \;|\; d\in \{2,3\}\}, \]
respectively.  
We set the associated subspaces 
        \[ \mathcal{A}^{(\leftrightarrow)}=\{\psi\in \mathcal{A} \;|\; a\notin A^{(\leftrightarrow)}\Rightarrow \psi(a)=0\},\;\;\mathcal{A}^{(\updownarrow)}=\{\psi\in \mathcal{A}\;|\; a\notin A^{(\updownarrow)}\Rightarrow \psi(a)=0\}. \]
We define unitary maps $\mathcal{U}$  and $\mathcal{U}_e$ which convert the
Hilbert space in the arc-set representations to ones in the vertex set representations with four and two internal states, respectively. 
These maps are given by $\mathcal{U}: \mathcal{A} \to \ell^2(V;\mathbb{C}^4)$ and 
$\mathcal{U}_e: \mathcal{A}^{(\leftrightarrow)} \to \ell^2(V;\mathbb{C}^2)$, where 
	\[ \ell^2(V;\mathbb{C}^n)=\{ f:V\to \mathbb{C}^n
	\;|\;\sum_{\bs{x}\in V} ||f(\bs{x})||^2_{\mathbb{C}^n} <\infty
	\} \;\; (n \in \{1,2,\dots\}) .\] 
More precisely, these maps $\mathcal{U}$ and $\mathcal{U}_e$ are defined as follows. 
\begin{definition}
\noindent \\
We define $\mathcal{U}: \mathcal{A} \to \ell^2(V;\mathbb{C}^4)$ such that 
	\[ (\mathcal{U}\psi)(\bs{x})
        	={}^T[ \psi(\bs{x};0)\;\; \psi(\bs{x};1)\;\; \psi(\bs{x};2)\;\; \psi(\bs{x};3) ].  \]
We also define $\mathcal{U}_e: \mathcal{A}^{(\leftrightarrow)} \to \ell^2(V;\mathbb{C}^2)$ such that 
	\[ (\mathcal{U}_e\psi)(\bs{x})
        	={}^T[ \psi(\bs{x};0)\;\; \psi(\bs{x};1)].  \]
\end{definition}
We can see that these inverse maps $\mathcal{U}^{-1}: \ell^2(V;\mathbb{C}^4)\to \mathcal{A}$ and $\mathcal{U}_e^{-1}: \ell^2(V;\mathbb{C}^2)\to \mathcal{A}^{(\leftrightarrow)}$ 
are expressed by 
	\begin{align*}
        (\mathcal{U}^{-1}f)(\bs{x};j) &= f_j(\bs{x})\;\;(j\in\{0,1,2,3\}),\\ (\mathcal{U}_e^{-1}g)(\bs{x};j) &= g_j(\bs{x})\;\;(j\in\{0,1\}) 
        \end{align*}
for $f\in \ell^2(V;\mathbb{C}^4)$ with $f(\bs{y})={}^T[f_0(\bs{y})\;f_1(\bs{y})\;f_2(\bs{y})\;f_3(\bs{y})]\in \mathbb{C}^4$ and 
$g\in \ell^2(V;\mathbb{C}^2)$ with $g(\bs{y})={}^T[g_0(\bs{y})\;g_1(\bs{y})]\in \mathbb{C}^2$.  
We put $|0\rangle={}^T[1,0]$ and $|1\rangle={}^T[0,1]$ as the standard basis of $\mathbb{C}^2$. 
Set the two-dimensional $\gamma$-rotation matrix by 
	\[ H_\gamma=\begin{bmatrix}  \cos \gamma & -\sin \gamma \\ \sin\gamma & \cos \gamma \end{bmatrix}. \;\;(\gamma\in[0,2\pi))\]
We define 
	\[ P_\gamma=|0\rangle\langle 0|H_\gamma,\; Q_\gamma=|1\rangle\langle 1|H_\gamma,\; S_\gamma=|1\rangle\langle 0|H_\gamma. \]
For $\psi\in \ell^2(V;\mathbb{C}^4)$ with $\psi(\bs{x})={}^T[\psi_0(\bs{x})\; \psi_1(\bs{x})\; \psi_2(\bs{x})\; \psi_3(\bs{x})]\in \mathbb{C}^4$, we denote 
$\psi^{(\leftrightarrow)}(\bs{x})={}^T[\psi_0(\bs{x})\; \psi_1(\bs{x})\;0\;0]$ and $\psi^{(\updownarrow)}(\bs{x})={}^T[0\;0\;\psi_2(\bs{x})\; \psi_3(\bs{x})]$. 
The arc representation of the time evolution of the quantum walk whose Hilbert space was $\ell^2(A)$ is converted to the vertex representation 
whose Hilbert space is $\ell^2(V;\mathbb{C}^4)$: 
\begin{lemma}\label{alternative}
Denote $U'=\mathcal{U}U\mathcal{U}^{-1}$. Then we have 
	\begin{align} 
        (U'\psi^{(\leftrightarrow)})(x,y) &= \tilde{Q}_\beta\psi^{(\leftrightarrow)}(x,y-1) + \tilde{P}_\beta\psi^{(\leftrightarrow)}(x,y+1) \label{even}\\
        (U'\psi^{(\updownarrow)})(x,y) &= 
        	\begin{cases}
                \tilde{Q}_\alpha\psi^{(\updownarrow)}(x-1,y) +\tilde{P}_\alpha\psi^{(\updownarrow)}(x+1,y) & \text{: $(x,y)\notin \partial V$,}\\
                \tilde{S}_\alpha\psi^{(\updownarrow)}(x,y) +
		 \tilde{P}_\alpha\psi^{(\updownarrow)}(x+1,y) & \text{:
		 $(x,y) \in \partial V$.}
                \end{cases} \label{odd}
        \end{align}
Here 
	\[ \tilde{P}_\alpha:=|0\rangle\langle 1|\otimes P_\alpha,\; \tilde{Q}_\alpha:=|0\rangle\langle 1|\otimes Q_\alpha,\; 
           \tilde{P}_\beta:=|1\rangle\langle 0|\otimes P_\beta,\;
	 \tilde{Q}_\beta:=|1\rangle\langle 0|\otimes Q_\beta,\;
\tilde{S}_\alpha:=|0\rangle\langle 1|\otimes S_\alpha.
 \]
\end{lemma}
\begin{proof}
See \ref{pfalternative}.
\end{proof}
Note that if $\gamma=\gamma'$, then $\tilde{X}_\gamma \tilde{X'}_{\gamma'}=0$ $(X,X'\in\{ P,Q \})$.
This property reflects the quantum walker's alternation between moving in the vertical and horizontal directions. 
Therefore, we can easily observe that 
	\[ U^2(\mathcal{A}^{(\leftrightarrow)})\subset \mathcal{A}^{(\leftrightarrow)},\;\; U^2(\mathcal{A}^{(\updownarrow)})\subset \mathcal{A}^{(\updownarrow)}. \]
Since our initial condition is $\psi_0=\delta_{(\bs{0};1)}\in \mathcal{A}^{(\leftrightarrow)}$ which is the self-loop at the origin, we concentrate on $\mathcal{A}^{(\leftrightarrow)}$. 
Using Lemma~\ref{alternative}, we obtain Lemma~\ref{alternative2}; see the detailed proof in \ref{pfalternative2}. 
	\begin{lemma}\label{alternative2}
        Let $\Gamma: \ell^2(V;\mathbb{C}^2)\to \ell^2(V;\mathbb{C}^2)$
	 be a map defined by for $\varphi\in \ell^2(V;\mathbb{C}^2)$, 
        \begin{multline}
        (\Gamma \varphi)(x,y)
        	= Q_\alpha Q_\beta\varphi(x-1,y-1)+Q_\alpha P_\beta\varphi(x-1,y+1) \\ +P_\alpha Q_\beta\varphi(x+1,y-1)+P_\alpha P_\beta\varphi(x+1,y+1),\;\;((x,y)\notin \partial V),
\label{eq:Gamma_D}
        \end{multline}
and
        \begin{multline}
        (\Gamma\varphi)(x,y)
        	= S_\alpha Q_\beta\varphi(x,y-1)+S_\alpha P_\beta\varphi(x,y+1) \\ +P_\alpha Q_\beta\varphi(x+1,y-1)+P_\alpha P_\beta\varphi(x+1,y+1),\;\;((x,y)\in \partial V).
\label{eq:Gamma_D2}
        \end{multline}
	Then we have 
        \[
	U^2|_{\mathcal{A}^{(\leftrightarrow)}}=\mathcal{U}_e^{-1}\Gamma
	\mathcal{U}_e. \]
        \end{lemma}
\begin{proof}
See \ref{pfalternative2}
\end{proof}
The notion of the weights associated with the one-step move to the neighboring vertex; $P$,$Q$ and $S$, 
comes from a quantum analog of a discrete-time random walk~\cite{Gudder} and becomes useful for the Fourier analysis in the next section. 
From now on, using Lemma~\ref{alternative}, we will focus upon the unitary operator $\Gamma$ on $\ell^2(V;\mathbb{C}^2)$ instead of the unitary operator $U$ on $\mathcal{A}$. 
\section{Connecting with the CMV matrix}
In the previous section, the model is reduced to the quantum walk on $\ell^2(V;\mathbb{C}^2)$ 
in which the walker moves to a neighboring location with the two-dimensional matrix weight following the unitary time evolution $\Gamma$. 
One of the keys to this paper is the connection between the AE model and the CMV matrix~\cite{CGMV}. 
The spectral theory on the CMV matrix can be used to obtain the
detailed spectral information, including gaps in the bulk spectrum and the edge state 
of our quantum walk model. 
To show this connection, we prepare the following two maps.
Set $L^2(\mathbb{Z}_+\times [0,2\pi);\mathbb{C}^2)$ by $\{ \hat{\varphi}:\mathbb{Z}_+\times [0,2\pi)\to \mathbb{C}^2 \;|\; \sum_{j\in \mathbb{Z}_+}\int_{0}^{2\pi} ||\hat{\varphi}(j;k)||^2_{\mathbb{C}^2} dk < \infty \}$.  
\begin{definition}
Let $\mathcal{F}: \ell^2(V;\mathbb{C}^2)\to L^2(\mathbb{Z}_+\times [0,2\pi);\mathbb{C}^2)$ be the Fourier transform defined by 
	\[ \hat{\varphi}(j;k):=(\mathcal{F}\varphi)(j;k)=\sum_{m\in
	\mathbb{Z}} \varphi(j,m)e^{ikm}\;\;(j\in\mathbb{Z},\; k\in
	[0,2\pi)), \]
and for fixed $k\in[0,2\pi)$ we can write $\varphi^\prime(\cdot)=\hat{\varphi}(\cdot\; ;k)$. 
\end{definition}
\begin{definition}
For fixed $k\in[0,2\pi)$, let $\Lambda_k: \ell^2(\mathbb{Z}_+;\mathbb{C}^2) \to \ell^2(\mathbb{Z}_+)$ be defined by 
	\[ (\Lambda_k \varphi')(j)=e^{i\omega (j)} \times 
        \begin{cases} 
        \langle 1| \varphi'([j/2]) \rangle  & \text{: $j$ is even,} \\ 
        \langle 0| \varphi'([j/2]) \rangle  & \text{: $j$ is odd.}
        \end{cases}
        \]
Here $[a]$ is the maximum integer such that $[a]\leq a$ for $a\in \mathbb{R}$ and 
	\[ \omega(2j)=-j\;\mathrm{arg}(\langle 0| \hat{H}_k |0 \rangle),\;\;\omega(2j+1)=(j+1)\;\mathrm{arg}(\langle 1| \hat{H}_k |1 \rangle), \]
where $\hat{H}_k \in \mathrm{SU}(2)$ is 
	\[ \hat{H}_k=\begin{bmatrix} 
        e^{-ik}\cos\alpha\cos\beta-e^{ik}\sin\alpha\sin\beta & -e^{-ik}\cos\alpha\sin\beta-e^{ik}\sin\alpha\cos\beta \\ 
        e^{-ik}\sin\alpha\cos\beta+e^{ik}\cos\alpha\sin\beta & -e^{-ik}\sin\alpha\sin\beta+e^{ik}\cos\alpha\cos\beta
        \end{bmatrix}.  \] 
\end{definition}
\begin{remark}
The inverse maps of $\mathcal{F}^{-1}: L^2(\mathbb{Z}_+\times[0,2\pi);\mathbb{C}^2)\to \ell^2(V;\mathbb{C}^2)$ 
and $\Lambda_k^{-1}: \ell^2(\mathbb{Z}_+)\to \ell^2(\mathbb{Z}_+;\mathbb{C}^2)$ are 
	\[ (\mathcal{F}^{-1}\hat{\varphi})(x,y)=\int_{0}^{2\pi} \hat{\varphi}(x;k)e^{-iky}\frac{dk}{2\pi},  \]
        \[ (\Lambda_k^{-1}f)(j)=\begin{bmatrix} e^{-i\omega(2j+1)}f(2j+1) \\ e^{-i\omega(2j)}f(2j) \end{bmatrix}. \]
\end{remark}
\begin{remark}
For every $\hat{\varphi}\in L^2(\mathbb{Z}_+\times [0,2\pi);\mathbb{C}^2)$, by Fubini's theorem, 
\begin{enumerate}
\item $\Lambda_k \hat{\varphi}$ is well-defined, because $\hat{\varphi}(\;\cdot\;;k)=\varphi'(\cdot)$ belongs to $\ell^2(\mathbb{Z}_+;\mathbb{C}^2)$ for fixed $k\in[0,2\pi)$;  
\item for every unitary operator $E$ on $\ell^2(\mathbb{Z}_+)$, we have $\Lambda_k^{-1}E \Lambda_k \hat{\varphi}(\;\cdot\;;k)\in \ell^2(V;\mathbb{C}^2)$ for fixed $k$,  
and $\Lambda_k^{-1}E \Lambda_k \hat{\varphi}\in L^2(\mathbb{Z}_+\times[0,2\pi);\mathbb{C}^2)$.
\end{enumerate}
\end{remark}
%
\subsection{The CMV matrix}
We use $\mathbb{T}$ to denote the unit circle on the complex plane, that is, $\mathbb{T}=\{z\in \mathbb{C}\;|\; |z|=1\}$. 
For a given positive measure $\mu$ on $\mathbb{T}$, we set the Hilbert space $L^2_\mu(\mathbb{T})$ whose inner product is defined by 
	\[ \langle f,g \rangle_{\mu}=\int_{z\in \mathbb{T}} \overline{f(z)}g(z) d\mu(z).  \]
Let $\{\chi_j(z)\}_{j=0}^{\infty}\subset L^2_\mu(\mathbb{D})$ be the orthogonal Laurent polynomials associated with the spectral measure $d\mu(z)$ on the unit circle in the complex plane 
obtained by the orthogonalization of $\{1,z,z^{-1},z^{2},\dots\}$.  The CMV matrix $\mathcal{C}$ is expressed by 
	\[ (\mathcal{C})_{i,j} = \langle \chi_i,z\chi_j \rangle_\mu \]
which represents a canonical representation of the multiplication operation on $L^2_\mu(\mathbb{T})$; $f(z)\mapsto zf(z)$. 
There is a one-to-one correspondence between the CMV matrix and so called Verblunsky parameters $(\eta_0,\eta_1,\eta_2,\dots)$, 
which are complex parameters satisfying $|\eta_j|<1$. This correspondence is given in \cite{CMV1,CMV2}. 
The CMV matrix is expressed in full as
\begin{equation*}
\mathcal{C}=
\begin{bmatrix}
\overline{\eta}_0      & \rho_0\overline{\eta}_1    & \rho_0\rho_1  & 0               & 0             & 0               & 0            & 0      & \ldots \\
\rho_0         & -\eta_0\overline{\eta}_1      & -\eta_0\rho_1    & 0               & 0             & 0                & 0             & 0       & \ldots \\
0              & \rho_1\overline{\eta}_2    & -\eta_1\overline{\eta}_2 & \rho_2\overline{\eta}_3 & \rho_2\rho_3  & 0                & 0             & 0       & \ldots \\  
0              & \rho_1\rho_2       & -\eta_1\rho_2    & -\eta_2\overline{\eta}_3   & -\eta_2\rho_3    & 0               & 0            & 0       & \ldots \\
0              & 0                  & 0             & \rho_3\overline{\eta}_4 & -\eta_3\overline{\eta}_4 & \rho_4\overline{\eta}_5 & \rho_4\rho_5 & 0       & \ldots \\
0               & 0                   & 0              & \rho_3\rho_4    & -\eta_3\rho_4    & -\eta_4\overline{\eta}_5   & -\eta_4\rho_5   & 0      & \ldots  \\
               &\vdots              &\vdots         &\vdots           & \vdots        & \vdots          & \ddots       &  &         
\end{bmatrix},
\end{equation*}
where $\rho_j=\sqrt{1-|\alpha_j|^2}$. 

Now we are ready to state the following theorem 
which claims that 
the time evolution of the AE model converted to the vertex-based expression; 
$\Gamma=\mathcal{U}_eU^2|_{\mathcal{A}^{(\leftrightarrow)}}\mathcal{U}_e^{-1}$, is decomposed into the CMV matrices $\{\mathcal{C}_k\}_{k=0}^{2\pi}$. 
        \begin{theorem}
        Let $\mathcal{C}_k$ be the CMV matrix whose Verblunsky parameters are $(\eta(k),0,\eta(k),0,\dots)$, where $\eta(k)=\sin(\alpha-\beta)\cos k+i\sin(\alpha+\beta)\sin k$. 
        Then we have  
        	\begin{equation}\label{connectionCMV}
                (\Gamma^{n}\varphi)(x,y)=\int_{0}^{2\pi}  \left( \Lambda_k^{-1}\; ({}^T\mathcal{C}_k)^n\; \Lambda_k\hat{\varphi}\right)(x)\;e^{-iky} \frac{dk}{2\pi},\;\;((x,y)\in V).
                \end{equation}
        \end{theorem}
The RHS of (\ref{connectionCMV}) is the $n$-th iteration value of the AE model at the vertex $(x,y)\in V$ with the initial state $\varphi\in \ell^2(V;\mathbb{C}^2)$, 
On the other hand, the integral on the LHS consists of the $2x+1$ and $2x$ components with the phase rotations $e^{-i(\omega(2 x+1)+ky)}$ and 
$e^{-i(\omega(2x)+ky)}$, respectively, 
of the $n$ th composition of the CMV matrix ${}^T\mathcal{C}_k$ with the initial state $\Lambda_k\hat{\varphi}\in \ell^2(\mathbb{Z}_+)$. 
Thus the problem is reduced to the analysis on the RHS of (\ref{connectionCMV}), 
using well developed studies on the spectral analysis on the CMV matrix to obtain spectral information about this walk. 

\subsection{Reduction to a quantum walk on the half line: Type-I quantum walk}
Due to the translation invariance to the $y$-axis direction of the AE model, we perform the partial Fourier transform defined in the previous subsection. 
We will show that the AE model is decomposed into a wave number dependent quantum walk on the half line. 
This reduced quantum walk is given by a vertex-based expression. 
For fixed $k\in [0,2\pi)$, we define $\hat{\Gamma}_k:\ell^2(\mathbb{Z};\mathbb{C}^2)\to \ell^2(\mathbb{Z};\mathbb{C}^2)$ by 
	\begin{equation}\label{Gammahat}
        (\hat{\Gamma}_k\varphi')(x) = 
        \begin{cases} 
        \hat{P}_k \varphi'(x+1)+\hat{Q}_k \varphi'(x-1) & \text{: $x\geq 1$, } \\
        \hat{P}_k \varphi'(x+1)+\hat{S}_k \varphi'(x) & \text{: $x=0$, }
        \end{cases}
        \end{equation}
where $\hat{P}_k=|0\rangle\langle 0|\hat{H}_k$, $\hat{Q}_k=|1\rangle\langle 1|\hat{H}_k$ and $\hat{S}_k=|1\rangle\langle 0|\hat{H}_k$. 
Here $\hat{H}_k$ is defined by $H_\alpha \hat{D}(k) H_\beta$, where 
	\[ \hat{D}_k=\begin{bmatrix} e^{-ik} & 0 \\ 0 & e^{ik} \end{bmatrix}. \]
More precisely, 
	\begin{equation}\label{quantumcoin} 
        \hat{H}_k=\begin{bmatrix} 
        e^{-ik}\cos\alpha\cos\beta-e^{ik}\sin\alpha\sin\beta & -e^{-ik}\cos\alpha\sin\beta-e^{ik}\sin\alpha\cos\beta \\ 
        e^{-ik}\sin\alpha\cos\beta+e^{ik}\cos\alpha\sin\beta & -e^{-ik}\sin\alpha\sin\beta+e^{ik}\cos\alpha\cos\beta
        \end{bmatrix}. 
        \end{equation}
Then we have the following lemma. 
	\begin{lemma}\label{lemma3}
        Let $\varphi_n\in \ell^2(V;\mathbb{C}^2)$ be the $n$-th iteration of $\Gamma$ with the initial state $\varphi_0(\bs{x})=\delta(\bs{x}){}^T[0,1]$, 
	that is, $\varphi_n=\Gamma^n\varphi_0$.
        The Fourier transform of $\varphi_n$ described by $\mathcal{F}(\varphi_n)=\hat{\varphi}_n \in L^2(\mathbb{Z}\times [0,2\pi);\mathbb{C}^2)$ is expressed by 
	\begin{equation}\label{psi}
        \hat{\varphi}_{n}=\hat{\Gamma}_k^n\hat{\varphi}_0,\;\hat{\varphi}_0=\delta(x){}^T[0,1]. 
        \end{equation}
        \end{lemma}
\begin{proof}
By the shift invariance with respect to the vertical direction and $Lemma~\ref{alternative2}$, we immediately prove this lemma. 
\end{proof}

The connection between the CMV matrix and the quantum walk was first shown in \cite{CGMV}. 
The Type-I quantum walk called by \cite{KS} is a quantum walk on the line induced by a special choice of the Verblunsky parameters of the CMV matrix. 
The Type-I quantum walk is defined by the arc-based expression as is the AE model's arc-based expression in Sect.~\ref{2.2}. 
Thus we prepare the notation for the arcs on the half line. 
Let $A'$ be the arc set of $\mathbb{Z}_+$ with the self-loop at the origin, that is, 
	\[ A'=\{ (x;0),(x;1) \;|\; x\in \mathbb{Z}_+ \}. \]
Here $(x;0)$ indicates the arc from $x+1$ to $x$, and $(x;1)$ indicates the arc from $x-1$ to $x$ for $x\geq 1$, and the self-loop for $x=0$. 
We define the subsets of $A'$ by 
	\[ A_x'=\{(x;0),(x;1)\} \]
for each $x\in \mathbb{Z}_+$. 
\begin{definition}
The Type-I quantum walk~\cite{CGMV, KS} is defined as follows: 
\begin{enumerate}
\item Total state space: $\mathcal{A}'=\ell^2(A')$
\item Time evolution: 
We set $\mathcal{A}'_x=\{ \varphi'\in \mathcal{A}'\;|\; a\notin A_x\Rightarrow \varphi'(a)=0 \}$. 
The local unitary operator on $\mathcal{A}'_j$ is defined by 
	\begin{align}\label{local}
        C'_x \cong \begin{bmatrix}  \gamma & \delta \\ \alpha & \beta \end{bmatrix} \\
        \end{align}
setting the canonical basis of $\mathcal{A}'_x$ by $\{\delta_{(x;0)}, \delta_{(x;1)}\}$ in this order. 
The time evolution $W: \mathcal{A}'\to \mathcal{A}'$ is described by 
	\begin{equation}
        (W\psi)(a)=\sum_{b:t(a)=o(b)} \langle  \delta_{\bar{a}},  C'_x \delta_{b}\rangle \psi(b) \;\mathrm{for}\; a\in A' \;\mathrm{with}\; o(a)=x. 
        \end{equation} 
\end{enumerate}
\end{definition}
As in the previous section, we also define $\mathcal{U}': \ell^2(A')\to \ell^2(\mathbb{Z}_{+};\mathbb{C}^2)$ as follows: 
	\[ (\mathcal{U}'\psi)(x)={}^T[\psi(x;0),\psi(x;1)] \;\;(x\in \mathbb{Z},\;j\in\{0,1\}). \]
The inverse map is $({\mathcal{U}'}^{-1}\varphi)(x;j)=\varphi_j(x)$. 
Putting $\Gamma'=\mathcal{U}'W{\mathcal{U}'}^{-1}$, we take the $n$-th iteration of $\Gamma'$ 
as $\varphi'_n$ with the initial state $\varphi'_0$, that is, $\varphi'_n={\Gamma'}^n \varphi_0$. 
From simple observation, we have 
	\begin{equation}\label{automaton}
        \varphi'_{n+1}(x)=
        \begin{cases}
        Q\varphi'_n(x-1)+P\varphi'_n(x+1) & \text{: $x\geq 1$,} \\
        S\varphi'_n(x)+P\varphi'_n(x+1) & \text{: $x=0$.}
        \end{cases}
        \end{equation}
Here $P=|0\rangle\langle 0|H$, $Q=|1\rangle\langle 1|H$ and $S=|0\rangle\langle 1|H$ with 
	\[ H=\begin{bmatrix} \alpha & \beta \\ \gamma & \delta \end{bmatrix}\in \mathrm{U}(2). \]
Therefore, $\mathcal{U}'\hat{\Gamma}_k{\mathcal{U}'}^{-1}$ is identical
to the time evolution of the Type-I quantum walk in the case of $H=\hat{H}_k$ for fixed $k$. 
Mapping the time iteration of $\hat{\Gamma}_k$ to that of Type-I quantum walk has the following benefit. 
	\begin{lemma}\cite{CGMV,KS} \label{CMV}
	For a given Type-I quantum walk with the local unitary coin $C'_j$~(\ref{local}), 
        let $\mathcal{D}: \mathcal{A}' \to \ell^2(\mathbb{Z}_+)$ be\footnote{The inverse map $\mathcal{D}^{-1}: \ell^2(\mathbb{Z}_+)\to \mathcal{A}'$ is 
        	\[ (\mathcal{D}f)(m;j)=\begin{cases} f(2m+1) e^{-i(2m+1)\mathrm{arg}(a)} & \text{: $j=0$,} \\ f(2m) e^{i2m \mathrm{arg}(d)} & \text{: $j=1$. }  \end{cases} \]} 
        	\[ (\mathcal{D}\psi)(m)=\begin{cases} \psi((m+1)/2;0) e^{i\mathrm{arg}(a)(m+1)/2} & \text{: $m$ is odd,} \\ \psi(m/2;1) e^{-i\mathrm{arg}(d) m/2} & \text{: $m$ is even.} \end{cases} \]
	The time evolution of the Type-I quantum walk, $W$, is unitarily equivalent to the CMV matrix with null odd Verblunsky parameter $(\eta_0,0,\eta_2,0,\eta_4,0,\dots)$, 
        that is, 
	\[ W=\mathcal{D}^{-1} \;{}^T\mathcal{C}\; \mathcal{D}, \]
	where 
        	\begin{equation}
        	\eta_j=\begin{cases} \eta\Delta^{-(j+1)/2} & \text{: $j$ is even,} \\ 0 & \text{: otherwise,} \end{cases}
        	\end{equation}
	with $\eta=\Delta^{1/2}\bar{\beta}$,\;$\Delta=\alpha\beta-\gamma\delta$. 
	\end{lemma}

\noindent Now we are ready to prove Theorem~1. \\
{\bf Proof of Theorem~1}\\ 
Since $\mathrm{det}(H_\alpha)=\mathrm{det}(H_\beta)=\mathrm{det}(\hat{D}_k)=1$, we have $\Delta=\mathrm{det}(\hat{H}_k)=1$. 
Thus concerning our local quantum coin (\ref{quantumcoin}), the corresponding Verblunsky parameter is $\eta_j=\eta(k)$ ($j$ is even), where
	\begin{equation}\label{VP}
        \eta(k)=\overline{\langle 0|\hat{H}_k|1\rangle}=-\cos k\sin(\alpha+\beta)+i\sin k\sin(\alpha-\beta). 
        \end{equation}
By Lemma~\ref{CMV}, 
	\begin{equation}\label{unitaryequiv} 
        \mathcal{U}'\; \hat{\Gamma}_k \; {\mathcal{U}'}^{-1} =  \mathcal{D}^{-1}\; {}^T\mathcal{C}_k\; \mathcal{D}. 
        \end{equation}  
Remarking that $\Lambda_k=\mathcal{D}\mathcal{U}'$, we have 
	\begin{equation}\label{unitaryequiv2} 
        \Lambda_k^{-1}{}^T\mathcal{C}^n_k\Lambda_k\hat{\varphi}_0=\hat{\varphi}_n. 
        \end{equation}
Taking the Fourier inverse of the above, that is, 
	\[ \varphi(x,y)=\int_{0}^{2\pi} \hat{\varphi}_n(x;k)e^{-iky}\frac{dk}{2\pi}, \]
we obtain the desired conclusion. $\square$  
\section{Spectrum, Dispersion Relation, and Topological Phases}
In this section, we study spectral properties of the quantum walks
$\hat{\Gamma}_k$ and compare them with the result by topological
phases which are independently studied from the quantum walk without
boundaries. 
In order to connect with both results, we focus on the wave number
dependence's of the spectrum, which is called dispersion relations in the
language of physics. We see that the location of spectrum agree well
with the result by topological numbers.

\subsection{Spectrum and dispersion relation of the AE model}
In the previous section, we proved the decomposition of the time evolution of the AE model $\Gamma$ into the wave number dependent quantum walk on the half line $\hat{\Gamma}_k$. 
Furthermore it is shown that this quantum walk $\hat{\Gamma}_k$ is co-spectral to the CMV matrix $\mathcal{C}_k$. 
Thus the spectral analysis on the CMV matrix for fixed $k$ can be used to get the spectral information about the AE model.
Indeed we will provide a quite explicit formula for the spectrum of $\hat{\Gamma}_k$ 
by the spectrum $\sigma(\mathcal{C}_k)$ of $\mathcal{C}_k$. 
We, furthermore, re-express this result in terms of physics, {\it i.e.}, the
dispersion relation, which clarifies the wave number dependence of the spectrum. Actually the
dispersion relation is described by nothing but the subset of $[0,2\pi)^2$ such
that 
$\{ (k,\theta) \;|\; e^{i\theta}\in \sigma(\mathcal{C}_k), k\in[0,2\pi)\}$.  \\
Now let us define some important functions to describe the above statement more explicitly: 
	\begin{align}
        \rho(k) & := \sqrt{1-|\eta(k)|^2}=\sqrt{\cos^2(\alpha-\beta)-\sin2\alpha\sin2\beta\cos^2k}, \label{rh}\\
        m_0(k) & := m_0
        =\begin{cases}
        \frac{|\cos k\sin(\alpha+\beta)|}{\sqrt{1-\sin^2(\alpha-\beta)\sin^2k}} & \text{: $\rho(k)\neq 0$,} \\ 
        1 & \text{: $\rho(k)=0$,}
        \end{cases} \label{mass} \\
        \theta_0(k) & := \begin{cases}
        \arcsin(-\sin k\sin(\alpha-\beta)) & \text{: $\cos k \sin(\alpha+\beta)\leq 0$,} \\ 
        \pi- \arcsin(-\sin k\sin(\alpha-\beta)) & \text{: $\cos k \sin(\alpha+\beta)> 0$.}\end{cases} \label{theta_0}
        \end{align}
We obtain a detailed expression for the spectrum of $\mathcal{C}_k$ using the above values as follows. 
	\begin{lemma}\label{bulkedgelemma}
	For fixed $k\in[0,2\pi)$, the spectrum of $\mathcal{C}_k$ is decomposed into continuous spectrum $\sigma_c^{(k)}$ 
	and point spectrum $\sigma_p^{(k)}$, that is, $\sigma(\mathcal{C}_k)=\sigma_c^{(k)} \sqcup \sigma_p^{(k)}$, where
		\begin{align} 
        	\sigma_c^{(k)} &= \{ e^{i\theta} \;|\; |\cos\theta|\leq \rho(k)\} \\
        	\sigma_p^{(k)} &= 
                \begin{cases} 
                \{e^{i \theta_0}(k)\} & \text{: $\sin(\alpha-\beta)\neq 0$, $k\notin \{\pi/2,3\pi/2\}$, }  \\ 
                \emptyset & \text{: otherwise.} 
                \end{cases}
        	\end{align}
	Moreover we have 
        	\begin{equation}
                \sigma(U^2|_{\mathcal{A}^{(\leftrightarrow)}})=\bigcup_{k\in[0,2\pi)}\left( \sigma_c^{(k)} \sqcup \sigma_p^{(k)} \right). 
                \end{equation}
        \end{lemma}
\begin{proof}
The detailed spectrum of this CMV matrix has already been obtained by \cite{CGMV,KS}. We summarize partially the results as follows. 
	\begin{lemma}\cite{CGMV,KS}\label{spectrumCMV}
	The spectral measure of the CMV matrix with the Verblunsky parameter $(\eta,0,\eta,0,\dots)$ is described as follows. 
	\begin{enumerate}
	\item For $|\eta|=1$ case. 
		\begin{equation}
        	d\mu(\theta)=\delta(\theta+\arccos(\eta)). 
        	\end{equation} 
        Here $\delta(\cdot)$ is the Dirac delta function on $[0,2\pi)$
	\item For $|\eta|<1$ case. 
		\begin{equation}
        	d\mu(\theta)=w(\theta)\frac{d\theta}{2\pi}+m_0\delta(\theta-\theta_0).
        	\end{equation}
	Here $w(\theta)$ is the weight of the absolutely continuous part, given by
		\begin{equation}\label{acdensity}
        	w(\theta)=\frac{\sqrt{\rho^2-\cos^2\theta}}{|\sin\theta+\mathrm{Im}(\eta)|}\bs{1}_{\{|\cos\theta|<\rho\}}(\theta),\;\;(\rho=\sqrt{1-|\eta|^2})
        	\end{equation}
	and $m_0\in [0,1]$ is the mass at $\theta_0\in [0,2\pi)$, where
		\begin{equation}
        	m_0=\frac{|\mathrm{Re}(\eta)|}{\sqrt{1-\mathrm{Im}^2(\eta)}},
        	\end{equation}
        	\begin{equation}\label{eigenvalue}
        	\theta_0=
                \begin{cases}
                \arcsin(-\mathrm{Im}(\eta)) & \text{: $\mathrm{Re}(\eta)\geq 0$,}\\
                \pi-\arcsin(-\mathrm{Im}(\eta)) & \text{: $\mathrm{Re}(\eta)<0$.}
                \end{cases} 
        	\end{equation}
	\end{enumerate}
	\end{lemma}
Thus Lemma~\ref{spectrumCMV} immediately implies the conclusion by replacing $\eta$ into $\eta=\eta(k)$. 
	 \end{proof}
Defining $\theta_c(k)$ as $\arccos(\rho(k))$, that is,  
		\begin{equation}\label{acpart} 
                \theta_c(k)=\arccos \left(\sqrt{\cos^2(\alpha-\beta)-\sin 2\alpha\sin 2\beta \cos^2k}\right), 
                \end{equation}
we find that the continuous spectrum of $\mathcal{C}_k$ is $\{e^{i\theta} \;|\; \theta\in [\theta_c(k),\pi-\theta_c(k)]\cup [\pi+\theta_c(k),2\pi-\theta_c(k)]\}$, 
and the point spectrum is $\{\theta_0(k)\}$.
In the language of physics, an argument of a spectrum of unitary
time-evolution operators is called quasi-energy, which is a main reason to
re-express $\sigma_c^{(k)}$ and $\sigma_p^{(k)}$ by quasi-energy
$\theta_c$ in (\ref{acpart}) and $\theta_0$ in (\ref{theta_0}), respectively. 
We call 
$[\theta_c(k),\pi-\theta_c(k)]\cup [\pi+\theta_c(k),2\pi-\theta_c(k)]$ 
and 
$\{\theta_0\}$ 
the bulk spectrum and the edge spectrum. 
These names come from the fact that the bulk spectrum corresponds to the
spectrum of walker's states in absence of any boundary (cutting edge of
systems), while the point spectrum is a consequence of the presence of
the cutting edge. As we show in section 4.2, the
existence of edge spectrum is consistent with a result by the bulk-edge
correspondence from a topological properties of the bulk states. 
Furthermore, 
we will show that states for the edge spectrum exponential localizes around the cutting edge in Section~6. 
Particularly, the wave number $k$ dependence of (quasi-)energy $\theta$ is called dispersion relation.
Thereby, $Bu$ and $Ed$ defined by (\ref{bu}) and (\ref{ed}) in the following definition can be
regarded as the dispersion relations. 
\begin{definition}\label{bulkedgedef}
Let $\sigma^{(k)}_c$ and $\sigma^{(k)}_p$ be the continuous and point spectrum's of $\mathcal{C}_k$ for fixed $k$. 
Then we set
	\begin{align}
        Bu &= \bigcup_{k\in [0,2\pi)} \{(k,\theta)\;|\; e^{i\theta}\in \sigma^{(k)}_c \}, \label{bu}\\
        Ed &= \bigcup_{k\in [0,2\pi)} \{(k,\theta)\;|\; e^{i\theta}\in \sigma^{(k)}_p \}. \label{ed}
        \end{align}
\end{definition}
In our system, the dispersion relation
is decomposed into $Bu$ and $Ed$, $Bu \cup Ed$, 
and we call $Bu$ and $Ed$ dispersion relations for bulk and edge
states, respectively.
The dispersion relation provides important information for
group and phase velocities, topological numbers, and so on.
By Lemma~\ref{bulkedgelemma}, we obtain quite explicit form of the
continuous spectrum $\sigma_c^{(k)}$ and the point spectrum
$\sigma_p^{(k)}$ of the wave number
dependent quantum walk on the half line $\hat{\Gamma}_k$ whose
spectrum is equivalent to  $\sigma(\mathcal{C}_k)$. 
The dispersion relation for bulk and edge states, $Bu$ and $Ed$, are described by $\sigma_c^{(k)}$ and $\sigma_p^{(k)}$. 
Therefore we can also provide explicit expressions for $Bu$ and $Ed\subset [0,2\pi)^2$ as follows.  
	\begin{theorem}
        Let $\theta_0(k)$ and $\theta_c(k)$ be the above defined in
	 (\ref{theta_0}) and (\ref{acpart}), respectively. Then we have 
	\begin{align} 
        Bu &= \bigcup_{k\in[0,2\pi)} [\theta_c(k),\pi-\theta_c(k)] \cup [\pi+\theta_c(k),2\pi-\theta_c(k)] \label{bulk} \\
        Ed &= \begin{cases}
                \{ (k,\theta_0(k)) \;|\; k\in[0,2\pi)\setminus \{\pi/2,3\pi/2\} \} & \text{: $\sin(\alpha-\beta)\neq 0$, } \\ 
                \emptyset & \text{: $\sin(\alpha-\beta)=0$, }
              \end{cases}  \label{edge}
        \end{align}
        \end{theorem}
By (\ref{bulk}), we can identify the condition for the
parameters $(\alpha,\beta)$ of the AE model so that there are no gaps in the bulk
spectrum.
	\begin{corollary}\label{corbulk}
        Let $\theta_0(k)$ and $\theta_c(k)$ be the above defined in (\ref{theta_0}) and (\ref{acpart}), respectively. 
	There exists wave number $k\in [0,2\pi)$ at which gaps in
	 the bulk spectrum vanish if and only if 
	``$\sin(\alpha-\beta)=0$ and $\sin 2\alpha\sin 2\beta\leq 0$" or ``$\sin(\alpha+\beta)=0$ and $\sin 2\alpha\sin 2\beta>0$". 
	In the first parameter class, the wave numbers are $k=\pi/2$ and $3\pi/2$, in the second parameter class, ones are $k=0$ and $\pi$. 
	\end{corollary}



We define $\theta_0(k)$ if the mass point $m_0(k)$ is strictly positive. 
Concerning the value of the mass point $m_0(k)$ described by (\ref{mass}), 
the parameter $(\alpha,\beta)$ must satisfy $\sin(\alpha+\beta)\neq 0$ and under this restriction, we set the domain of $\theta_0(k)$ as 
$[0,2\pi)\setminus \{\pi/2,3\pi/2\}$. 
Under this setting of the parameter and the domain, we call
$\theta_0(k)$ the quasi-energy for edge states or the edge
spectrum, and $Ed$ the dispersion relation for edge states. Furthermore, the derivative
of the quasi-energy for edge states $\theta_0(k)$ is called
the group velocity for the edge state. (Since we focus on the group
velocity only for the edge states in the present work, we omit ``for the edge
states'' hereafter to use the term.) The group velocity corresponds to the
velocity of the envelope of a wave packet, as we will explain in section 5.
        \begin{corollary}\label{coredge}
        The dispersion relation for edge states is monotonic on $\mathbb{T}:=\mathbb{R}/(2\pi\mathbb{Z})$ 
        with $Ed\cap Bu=\emptyset$, and 
        is discontinuous only at $k=\pi/2$ and $3\pi/2$ $\mathrm{mod}(2\pi)$ on $\mathbb{T}$. At these points there is a jump, that is,  
        	\[ \lim_{k\uparrow \pi/2} \theta_0(k) \neq
		\lim_{k\downarrow \pi/2} \theta_0(k),\;\;\lim_{k\uparrow
		3\pi/2} \theta_0(k) \neq \lim_{k\downarrow 3\pi/2}
		\theta_0(k), \]
         and 
	\[ \lim_{k\uparrow \pi/2,3\pi/2} \theta_0(k),\;\lim_{k\downarrow \pi/2,3\pi/2}\theta_0(k)\in Bu. \]
         This relation explains why we relate $\theta_0$ and $Ed$ with
	 edge states, since it is well known in the research field of
	 topological phases of matter that edge states
	 originating from topological phases should appear in gaps of the
	 bulk spectrum so that the edge spectrum connects two endpoints
	 of the different interval of the bulk spectrum. Furthermore, in
	 section 6, we provide a clear evidence that states with
	 $\theta_0$ exponentially localize near self-loops on the
	 boundary by deriving the limit distribution towards bulk.

	More precisely, we divide the space of parameters $(\alpha,\beta)\in [0,2\pi)^2$ 
        into six classes according to the pair of signs; $(\mathrm{sgn}(\sin(\alpha+\beta)), \mathrm{sgn}(\sin(\alpha-\beta)))$, as follows: 
        \begin{enumerate}
        \item $(0,0)$ and $(0,\pm )$ classes: \\
        This is the condition for vanishing gaps in the bulk spectrum at $k=0$ and
	      $k=\pi$.
	There are no edge states for all $k\in[0,2\pi)$, that is, $\theta_0(k)$ cannot be defined.  
        \item $(+,+)$ class: \\ 
        The dispersion relation for edge states monotonically increases on $\mathbb{T}$. 
        The discontinuity at $k=\pi/2$ occurs 
        between ``$(\pi/2,\pi+\theta_c(\pi/2))$ and $(\pi/2,2\pi-\theta_c(\pi/2))$" $\in [0,2\pi)^2$  
        and the one at $k=3\pi/2$ occurs between ``$(3\pi/2,\theta_c(3\pi/2))$ and $(\pi/2,\pi-\theta_c(3\pi/2))$" $\in [0,2\pi)^2$, respectively. 
        \item $(+,-)$ class: \\
        The dispersion relation for edge states monotonically decreases on $\mathbb{T}$. 
        The discontinuity at $k=\pi/2$ occurs 
        between ``$(\pi/2,\theta_c(\pi/2))$ and $(\pi/2,\pi-\theta_c(\pi/2))$" $\in [0,2\pi)^2$ 
        and that at $k=3\pi/2$ occurs between ``$(3\pi/2,\pi+\theta_c(3\pi/2))$ and $(\pi/2,2\pi-\theta_c(3\pi/2))$"$\in [0,2\pi)^2$, respectively. 
        \item $(-,+)$ class: \\
        The dispersion relation for edge states monotonically decreases on $\mathbb{T}$. 
        The discontinuity at $k=\pi/2$ occurs 
        between ``$(\pi/2,\pi+\theta_c(\pi/2))$ and $(\pi/2,2\pi-\theta_c(\pi/2))$" $\in [0,2\pi)^2$ 
        and that at $k=3\pi/2$ occurs between ``$(3\pi/2,\pi+\theta_c(3\pi/2))$ and $(\pi/2,2\pi-\theta_c(3\pi/2))$"$\in [0,2\pi)^2$, respectively. 
        \item $(-,-)$ class: \\
        The dispersion relation for edge states monotonically increases on $\mathbb{T}$. 
        The discontinuity at $k=\pi/2$ occurs 
        between ``$(\pi/2,\theta_c(\pi/2))$ and $(\pi/2,\pi-\theta_c(\pi/2))$" $\in [0,2\pi)^2$ 
        and that at $k=3\pi/2$ occurs between ``$(3\pi/2,\pi+\theta_c(3\pi/2))$ and $(\pi/2,2\pi-\theta_c(3\pi/2))$"$\in [0,2\pi)^2$, respectively. 
        \item $(\pm ,0)$ class: \\
        This is the condition for vanishing gaps in the bulk spectrum at $k=\pi/2$ and $k=3\pi/2$. 
        The group velocity is zero, that is, the quasi-energy for the edge states has no wave number dependence.
        \end{enumerate}
        \end{corollary}

 By the above corollaries, the pattern shape of the dispersion relation $Bu\cup Ed$ is determined by the 
signs of $\sin 2\alpha\sin2\beta=\sin^2(\alpha+\beta)-\sin^2(\alpha-\beta)$, $\sin(\alpha+\beta)$ and 
$\sin(\alpha-\beta)$. See Figure 2. 
We set the triple of the signs $(\epsilon_1,\epsilon_2,\epsilon_3)$, $(\epsilon_j\in\{\pm\})$, 
where $\epsilon_1=\sgn(\sin 2\alpha\sin2\beta)$, $\epsilon_2=\sgn(\sin(\alpha+\beta)$ and $\epsilon_3=\sgn(\sin(\alpha-\beta))$.  
The first sign provides the shape of $Bu$: we have two bands which are symmetric with respect to $\theta=\pi$, 
and each band has period $\pi$; 
if $\epsilon_1=+$, then ``peaks" appear at $k=\pi/2,3\pi/2$  and ``troughs" appear at $k=0,\pi$ 
while if $\epsilon_1=-$, then we have the $\pi$-phase difference of the $\epsilon_1=+$ case. 
If $\sin(\alpha+\beta)\neq 0$, the edge state appears as a bridge connecting the peaks (resp. troughs ) of the upper band 
and the lower band for $\epsilon_1=+$ case (resp. $\epsilon_1=-$). 
The signs $\epsilon_2$ and $\epsilon_3$ describe the pattern shape of the edge state. 
If $\epsilon_2=+$, then $\theta_0(0)=\pi$ and $\theta_0(\pi)=0(=2\pi)$, while if $\epsilon_2=-$, then we have the $\pi$-phase shift of $\epsilon_2=+$ case. 
If $\epsilon_2\epsilon_3=+$, then the group velocity is positive, 
while if $\epsilon_2\epsilon_3=-$, then the group velocity is negative. 

\begin{figure}[htbp]
  \begin{center}
              \includegraphics[clip, width=14cm]{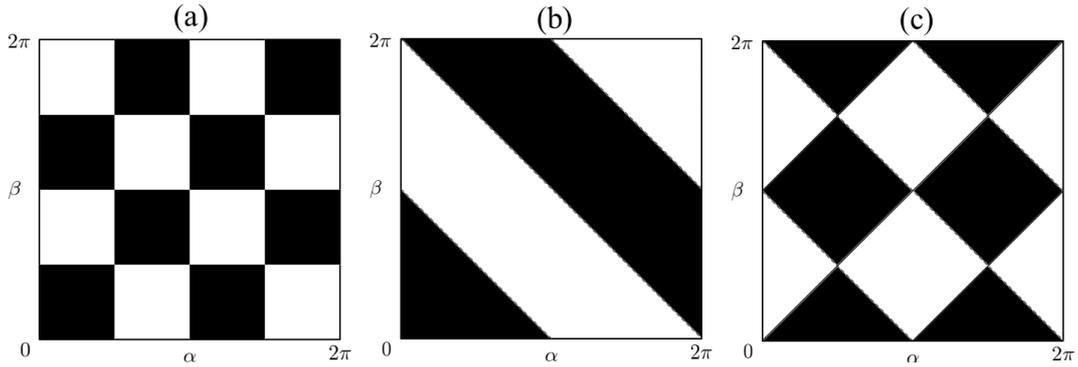}
              \caption{The black regions of $[0,2\pi)^2$ 
              depict (a) $\epsilon_1=\sin 2\alpha\sin2\beta>0$ (b) $\epsilon_2=\sin(\alpha+\beta)>0$ (c) $\epsilon_2\epsilon_3=\sin(\alpha+\beta)\sin(\alpha-\beta)>0$, respectively. 
              On the lines $\alpha+\beta=\pi,2\pi,3\pi$ (see i.g., (b)), the edge state does not appear. }
  \end{center}
\end{figure}

\subsection{Topological phases}
Here, we study symmetry and topological phases of the AE model. 
At first, we review basic notions of topological phases, then
explain symmetry and topological numbers of the AE model
studied in \cite{AE}. We remark that the topological  number of the
AE model derived in
\cite{AE} is well-defined only when gaps in the bulk spectrum exist.
As we will explain later, the result by obtained using the topological phases
agrees with the presence or absence of edge states by Corollary 2
(1)-(5), but not Corollary 2 (6). This is because 
Corollary 2 (6) requires the absence of gaps in the bulk spectrum, where
the topological numbers in \cite{AE} are not well-defined.
We solve this problem by introducing a new topological number of the AE
model, which is defined only when there are no gaps in the bulk spectrum.

Recently, localization appearing in space-inhomogeneous
quantum walks has
attracted attention, since it has been related to topological phase.
Here, the topological numbers originating from non-trivial topological
properties of the eigen functions in wavenumber space\cite{volovik07} are
related to fundamental objects of study in algebraic topology, such as
winding number or Chern number. This approach
has emerged from recent research of topological phases of matter (known
as topological insulators and superconductors) in condensed-matter
physics\cite{bernevig13, chiu15, hasan10, qi11}. 
Since a
quantum walk can be considered as a simplified
theoretical model of a topological insulator (more preciously, Floquet
topological insulators\cite{Kitagawa3,Lindner11}, since quantum
walks can be regarded as time-periodically driven systems), the topological
phases of quantum walks in inhomogeneous systems have been intensively studied~\cite{A,CGGSVWW,CGSVWW,EKO,EEKST,OK,OANK}. 
In the field of theoretical physics of the topological insulators, a fundamental principle,
the so called bulk-edge correspondence, states that
the number of edge states localized at the interface between two adjacent
regions in position space is the same as the difference in topological number between
the two regions.

Now, we briefly introduce results concerning topological numbers in \cite{AE}.
Since topological phases of matter are easier to consider for
systems with translation symmetry (though this is not necessary for
topological insulators), we focus
on a system without boundaries and  explain topological numbers originating
from non-trivial topological properties
of eigen functions in the wavenumber space. Accordingly, 
we consider the time-evolution operator $\Gamma$ in 
(\ref{eq:Gamma_D}) for the vertex set
	\[ V_0=\mathbb{Z} \times \mathbb{Z}=\{ (x,y)\;|\; x\in \mathbb{Z}, y\in \mathbb{Z} \}, \]
in this section. 
Unlike (\ref{Gammahat}),
in applying the Fourier transform to the two-dimensional space of $\Gamma$, 
we derive the time-evolution operator in the wavenumber representation 
$\hat{\Gamma}_0(k_x,k_y) : \mathbb{C}^2 \to
\mathbb{C}^2$ for fixed $(k_x,k_y) \in [0,2\pi)^2$
described as
\begin{equation}
 \hat{\Gamma}_0(k_x,k_y) = \hat{D}(k_x) H_\alpha  \hat{D}(k_y) H_\beta.
\label{eq:Gamma_2D}
\end{equation}
We note that $\hat{D}(k)$ and $H_\gamma$ can be written as
\begin{align*}
\hat{D}(k)&= e^{-ik \sigma_3},\\
H_\gamma & = e^{-i \gamma \sigma_2},
\end{align*}
respectively, 
using an identity matrix $\sigma_0$ and Pauli
matrices $\sigma_i$ $(i=1,2,3)$
\[
\sigma_0=
\begin{bmatrix}
1  & 0\\
0 & 1
\end{bmatrix} 
,
\sigma_1=
\begin{bmatrix}
0  & 1\\
1 & 0
\end{bmatrix} 
,
\sigma_2=
\begin{bmatrix}
0  & -i\\
i & 0
\end{bmatrix} 
,
\sigma_3=
\begin{bmatrix}
1  & 0\\
0 & -1
\end{bmatrix} .
\]

First, we identify the symmetries
of the time-evolution operator, $\hat{\Gamma}_0(k_x,k_y)$.
It is well known that there are three important symmetries for topological phases, namely, time-reversal, particle-hole, and
 chiral symmetries\cite{bernevig13, chiu15, hasan10, qi11}.
These symmetries require the time-evolution operator to satisfy
the following relations with anti-linear symmetry operators $\hat{T}$ and 
$\hat{P}$, and a linear symmetry operator $\hat{Y}$(see details in \cite{EKO}):
\begin{align}
\text{Time-reversal symmetry:}\quad\quad \hat{T} \hat{\Gamma}_0(k_x,k_y) \hat{T}^{-1} &= \hat{\Gamma}_0(-k_x,-k_y)^{-1},\label{eq:TRS}\\
\text{Particle-hole symmetry:}\quad\quad \hat{P} \hat{\Gamma}_0(k_x,k_y) \hat{P}^{-1} &= \hat{\Gamma}_0(-k_x,-k_y),\label{eq:PHS}\\
\text{Chiral symmetry:}\quad\quad \hat{Y} \hat{\Gamma}_0(k_x,k_y) \hat{Y}^{-1} &= \hat{\Gamma}_0(k_x,k_y)^{-1},\label{eq:chiral}
\end{align}
On one hand, we see from (\ref{eq:Gamma_2D}) and (\ref{eq:PHS}) that particle-hole symmetry is retained  with the symmetry operator
\[
 \hat{P} = K,
\]
where $K$ stands for the complex conjugation.
On the other hand, time-reversal and chiral symmetries
are not generally satisfied.
Therefore, the AE model belongs to class D in the classification
table of topological phases\cite{SRFL}. 

Roughly speaking,  topological phases of matter are characterized by topological properties of
wavefunctions, which are periodic functions in wavenumber space. 
The topological number $\nu_{2d}$ for the AE model derived in
Ref.\ \cite{AE} is obtained from
a winding number originating from the periodicity of two
wavenumber space $k_x$ and $k_y$ in addition to the time direction 
due to the periodic-time drive of quantum walks\cite{RLBL}.
The topological number $\nu_{2d}$'s dependence on $a$ and $b$
is summarized in figure \ref{fig:topological_numbers} (a). 
We note that the topological number $\nu_{2d}$ is finite except on lines in figure
\ref{fig:topological_numbers} (a), where there are no gaps in the
bulk spectrum and 
$\nu_{2d}$ is not well-defined.
Along with the bulk-edge correspondence, we can predict the emergence of
non-degenerated edge states in the AE model when the topological number $\nu_{2d}$ assigned to the parameter set
$(\alpha,\beta)$ is non-zero, since in the outer region from the cutting
edge of the AE model, $\nu_{2d}$ is zero. 
Since the topological number $\nu_{2d}$ is finite whenever there exist
gaps in bulk spectrum.
the phase diagram seems to agree with Corollary~\ref{coredge} (1)-(5).
\begin{figure}[tbp]
  \begin{center}
              \includegraphics[clip, width=17cm]{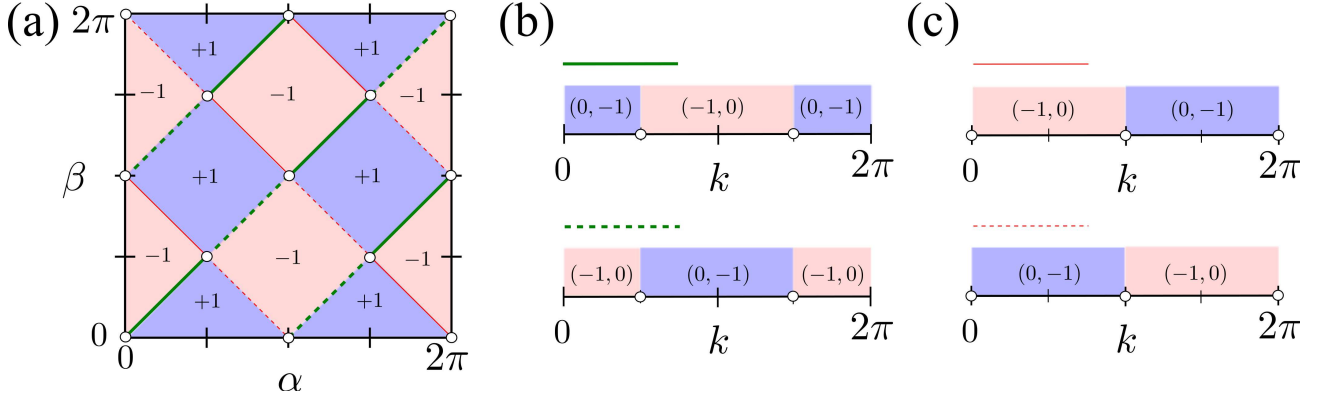}
              \caption{(a) The phase diagram of topological number
   $\nu_{2d}$ derived in Ref.\ \cite{AE}
 as functions of $\alpha$ and  $\beta$. In addition, the
   green thick and red thin (solid and dashed) lines represent 
   vanishing gaps in the bulk spectrum} at the quasienergies of $0$ and $\pi$, 
 where the topological numbers $(\nu_0,\nu_\pi)$ are given in (\ref{eq:toponum1}) and
   (\ref{eq:toponum2}). (b and c) The wavenumber $k$ dependence
   of the topological numbers $(\nu_0,\nu_\pi)$ on the green thick and red
   thin (solid and dashed) lines, respectively.
At the points shown by open
   circles in all figures, no
   topological numbers are defined. Note that the topological
   numbers on the red thin lines cannot be applied to the AE model. 
\label{fig:topological_numbers}
  \end{center}
\end{figure}

Corollary~\ref{coredge} (6), however, states that edge states
appear despite the fact that there are no gaps in the bulk spectrum and the
topological number is not well-defined. 
This disagreement can be resolved by introducing a topological number
that is calculated from $\hat{\Gamma}_0(k_x,k_y)$ with fixed wave-number $k_y$ such
that gaps in the spectrum in the subspace $k_x \in
[0,2\pi)$ exist\cite{SR}. Therefore, the new topological number is similar to
one defined in one dimension, as previously studied\cite{AO}. Remarkably, this topological
number is well-defined only when there are no gaps in the
two-dimensional bulk spectrum in order to retain chiral symmetry. This is the main result of this section, and we explain it hereafter.

First, we impose the following restrictions on $\alpha$ and $\beta$ in (\ref{eq:Gamma_2D}) to satisfy case (6) of Corollary~\ref{coredge}:
\begin{equation}
 \beta \in \{\alpha+n\pi\;|\; n\in\mathbb{Z},  \alpha \in[0,2\pi)
  \setminus \{0,\frac{\pi}{2}, \pi,\frac{3\pi}{2}\} \}.
\label{eq:alpha_beta_class6}
\end{equation}
Then, the time-evolution operator becomes
\begin{align*}
\hat{\Gamma}_0(k_x, k_y; n,\alpha) = (-1)^n \hat{\Gamma}_0(k_x,k_y;\alpha),\quad
 \hat{\Gamma}_0(k_x,k_y;\alpha) = \hat{D}(k_x) H_\alpha  \hat{D}(k_y) H_\alpha.
\end{align*}
Since the prefactor $(-1)^n=e^{in\pi}$ only shifts the quasi-energy by
$n\pi$, we can focus on $\hat{\Gamma}_0(k_x,k_y;\alpha)$ and recover the prefactor
at the end of the calculations.
Now, by shifting the origin of time to derive a time-evolution operator fitted in the
symmetry time frame\cite{AO}, we obtain
$\hat{\Gamma}_0^\prime(k_x,k_y;\alpha)$
\begin{align}
\hat{\Gamma}_0^\prime(k_x,k_y;\alpha) &=  \hat{D}(k_x/2) H_{\alpha} \hat{D}(k_y)
 H_\alpha  \hat{D}(k_x/2).
\label{eq:Gamma_2D_chiral}
\end{align}
We see that $\hat{\Gamma}_0^\prime(k_x,k_y;\alpha)$ satisfies the relation
\begin{align*}
\hat{Y} \hat{\Gamma}_0^\prime(k_x,k_y;\alpha) \hat{Y}^{-1} &=
\hat{\Gamma}_0^\prime(k_x,k_y;\alpha)^{-1},
\end{align*}
with chiral symmetry operator
\begin{equation}
 \hat{Y}= \sigma_1.
\label{eq:Y1}
\end{equation}
Note that the presence of chiral symmetry relies on the
condition (\ref{eq:alpha_beta_class6}) resulting in the absence of gaps
in the bulk spectrum, as explained by Corollary 2 (6).
Since $\hat{\Gamma}_0^\prime(k_x,k_y;\alpha)$ also retains particle-hole
symmetry, the system with parameters (\ref{eq:alpha_beta_class6}) belongs to the class BDI. 
While the class BDI in two dimensions does not have finite
topological numbers, chiral symmetry classes can have an integer topological number in
one dimension\cite{SRFL}. 
Therefore, we fix $k_y$ of the time-evolution operator
$\hat{\Gamma}_0^\prime(k_x,k_y;\alpha)$ in (\ref{eq:Gamma_2D_chiral}) and regard it as a
parameter, $k \in [0,2\pi)$. In the end, we consider the following
time-evolution operator,
\begin{align}
\hat{\Gamma}_0^\prime(k_x; k, n) = (-1)^n \hat{\Gamma}_0^\prime(k_x;k),\quad
 \hat{\Gamma}_0^\prime(k_x;k) = \hat{D}(k_x/2) H_\alpha
 \hat{D}(k) H_\alpha \hat{D}(k_x/2).
\label{eq:Gamma_1D_chiral}
\end{align}
Since $k$ is fixed, $\hat{\Gamma}_0^\prime(k_x;k)$ effectively describes 
time-evolution in one dimension with chiral symmetry, while $\hat{D}(k)$
introduces a complex factor $e^{-ik\sigma_3}$. Therefore, the system
belongs to the class AIII.
Because of the presence of chiral symmetry, the topological
number is calculated by the following the method in Ref.\ \cite{AO} as long as the
gaps in the spectrum in the subspace $k_x\in[0,2\pi)$ exist for a fixed $k$. 
In order to calculate the winding number, it is better to express
$\hat{\Gamma}_0^\prime(k_x;k)$ using Pauli matrices as
\begin{align*}
\hat{\Gamma}_0^\prime(k_x;k) &= d_0^\prime(k_x)\sigma_0 + i \sum_{j=1,2,3}
 d_j^\prime(k_x)\sigma_j,\nonumber \\
d_0^\prime(k_x) & = \cos k\cos(2\alpha) \cos k_x - \sin k \sin
 k_x,\nonumber \\
d_1^\prime(k_x) & = 0,\nonumber \\
d_2^\prime(k_x) & = d_2^\prime= -\sin (2\alpha) \cos k,\\
d_3^\prime(k_x) & = -\sin k \cos k_x - \cos k \cos(2\alpha) \sin k_x.
\end{align*}
Furthermore, we apply a unitary transformation to $\hat{\Gamma}_0^\prime(k_x;k)$ such
that the chiral symmetry operator becomes $\hat{Y}=\sigma_3$,
\begin{align}
\tilde{\Gamma}_0^\prime(k_x;k) & = e^{-i\frac{\pi}{4}\sigma_2}
\hat{\Gamma}_0^\prime(k_x;k) e^{i\frac{\pi}{4}\sigma_2},\nonumber \\
&=
d_0^\prime(k_x) \sigma_0 + id_3^\prime(k_x)\sigma_1 + i d_2^\prime(k_x)\sigma_2.\nonumber
\end{align}
The eigenvalue $\lambda_\pm$ and corresponding eigenvectors $\psi_\pm^\prime(k_x)$ of $\tilde{\Gamma}_0^\prime(k_x;k)$ are given by
\begin{align}
\lambda_\pm &= d_0^\prime(k_x) \pm i \sqrt{1-d_0^{\prime2}(k_x)},\nonumber\\
\psi_\pm^\prime(k_x)  &= \frac{1}{\sqrt{2}}
\begin{bmatrix}
\mp i e^{i \theta^\prime(k_x)} \\ 1
\end{bmatrix},\nonumber \\
e^{i \theta^\prime(k_x)} &=
 \frac{d_2^\prime(k_x)+id_3^\prime(k_x)}{\sqrt{1-d_0^{\prime 2}(k_x)}} \label{eq:theta}
\end{align}
Then, the winding number characterizing the topological properties of
the eigenvector in wavenumber space is given by
 \begin{align*}
 \nu^\prime &= -\frac{i}{\pi} \int_0^{2\pi} \overline{\psi_\pm^\prime(k_x)} \,\,
\frac{d \psi_\pm^\prime(k_x)}{dk_x}  dk_x\nonumber \\
&= \frac{1}{2}\oint d\theta^\prime(k_x).
 \end{align*}

Therefore, the winding number is determined by whether 
$e^{i\theta^\prime(k_x)}$ in (\ref{eq:theta}) forms a closed circle including the
origin in the complex plane when $k_x$ changes from $0$ to $2\pi$.
Since $\sqrt{1-d_0^{\prime 2}(k_x)}$ is positive finite as long as the spectrum
$\lambda_\pm$ has gaps, we focus only upon the numerator in 
(\ref{eq:theta}), $d_2^\prime(k_x)+id_3^\prime(k_x)$, in the following argument. 
Then, because  $d^\prime_2(k_x)$ is independent of $k_x$,
$e^{i\theta^\prime(k_x)}$ does not wind around the origin and 
we have 
\[
\nu^{\prime}=0 
\]
for any $\alpha$ and $k$.

According to \cite{AO}, in order to calculate the topological numbers $\nu_0$ and $\nu_\pi$ for
quasienergies $0$ and $\pi$, respectively, we need to calculate the
winding number for the other time-evolution operator $\hat{\Gamma}_0^{\prime\prime}(k_x;k)$ fitted in the
different symmetry time frame
\[
 \hat{\Gamma}_0^{\prime\prime}(k_x;k) = \hat{D}(k/2) H_{\alpha} \hat{D}(k_x)
 H_\alpha  \hat{D}(k/2).
\]
Because of the relation 
$\hat{\Gamma}_0^{\prime\prime}(k_x;k)=\hat{\Gamma}_0^\prime(k;k_x)$,
the winding number $\nu^{\prime\prime}$ of $\hat{\Gamma}_0^{\prime\prime}(k_x;k)$ is
given by switching $k_x$ and $k$ in the derivation of $\hat{\Gamma}_0^\prime(k;k_x)$. In
other words, the topological number $\nu^{\prime\prime}$ is 
given by
\begin{align}
 \nu^{\prime\prime} &= \frac{1}{2}\oint d\theta^{\prime\prime}(k_x),\nonumber \\
e^{i \theta^{\prime\prime}(k_x)} &=
 \frac{d_2^{\prime\prime}(k_x)+id_3^{\prime\prime}(k_x)}{\sqrt{1-d_0^{\prime\prime
 2}(k_x)}},\nonumber \\
d_2^{\prime\prime}(k_x) & = -\sin (2\alpha) \cos k_x,\nonumber \\
d_3^{\prime\prime}(k_x) & = -\cos k \sin k_x - \sin k \cos(2\alpha) \cos k_x.\nonumber
\end{align}
Putting $u=d_2^{\prime\prime}(k_x)$ and
$v=d_3^{\prime\prime}(k_x)$, we obtain the elliptic function of $u$ and $v$
\begin{equation*}
(\cos^2 k + \sin^2 k \cos^2(2\alpha)) u^2 + \sin^2(2\alpha) v^2 
-2\sin k \cos(2 \alpha) \sin(2\alpha) uv = \cos^2 k \sin^2(2\alpha)
\end{equation*}
By properly rotating the coordinate, the above elliptic function is
simply written as
\begin{align}
&u^{\prime2} + \cos^2 k \sin^2(2\alpha) v^{\prime 2}= \cos^2 k
 \sin^2(2\alpha),\label{eq:elliptic}\\
&
\begin{bmatrix}
u^\prime \\ v^\prime 
\end{bmatrix}
=
\begin{bmatrix}
\cos \eta & \sin \eta \\
-\sin \eta & \cos \eta
\end{bmatrix}
\begin{bmatrix}
u \\ v
\end{bmatrix},\nonumber \\
&\cos \eta =\frac{\cos(2\alpha)}{\Xi}, \quad 
\sin \eta  =\frac{\sin k \sin(2\alpha)}{\Xi},\quad
\Xi = \sqrt{1-\cos^2 k \sin^2 (2\alpha)}.\nonumber
\end{align}

From (\ref{eq:elliptic}), the winding number $\nu^{\prime\prime}$ is
determined as follows.
\begin{enumerate}
 \item The case of
       $\alpha=0,\frac{\pi}{2},\pi,\frac{3\pi}{2}$: (these values are excluded by (\ref{eq:alpha_beta_class6}))

Since (\ref{eq:elliptic}) is reduced to $u^\prime=u=0$, the
       trajectory of $\theta^{\prime\prime}(k_x)$ 
       becomes a line, and then $\nu^{\prime\prime}=0$.

 \item The case of $k=\frac{\pi}{2}, \frac{3\pi}{2}$:

Since (\ref{eq:elliptic}) is reduced to
       $u^{\prime}=\cos(2\alpha)u \pm \sin(2\alpha) v=0$, the trajectory of $\theta^{\prime\prime}(k_x)$ 
       becomes a line and then $\nu^{\prime\prime}=0$.

 \item Otherwise:
 
The trajectory of $\bs{p}=(u(k_x), v(k_x))$  $k_x\in [0,2\pi)$ 
       is a rotated ellipse enclosing the origin. 
The curvature of $\bs{p}$ is expressed by 
       \[ \kappa(k_x)=\frac{\dot{u}(k_x)\ddot{v}(k_x)-\ddot{u}(k_x)\dot{v}(k_x)}{(\dot{u}(k_x)^2+\dot{v}(k_x)^2)^{3/2}}. \]
Here, $\dot{w}=dw/dk_x$, and $\ddot{w}=d^2w/dk_x^2$, $(w=u,v)$. 
The numerator of $\kappa(k_x)$ in this case can be computed as $\sin 2\alpha \cos k$. 
Thus, the sign of $\kappa$ coincides with that of $\sin 2\alpha \cos k$, 
meaning $\nu^{\prime\prime}\ne 0$ if and only if $\sin 2\alpha \cos k \ne 0$. 
Moreover, since $\dot{u}(k_x)=\sin k_x \sin 2\alpha$, then $k_x=\pi$ is a unique point in $(0,2\pi)$ 
such that the sign of $\dot{u}(k_x)$ is inverted, that is, for $\sin{2\alpha}>0$
	\[ \dot{u}(k_x)>0\; (k_x\in (0,\pi)),\;\; \dot{u}(k_x)<0\;
	(k_x\in (\pi,2\pi)),\] 
and vice versa for $\sin{2\alpha}<0$.
Thus since the orbit is a rotated ellipse enclosing the origin, we have $|\nu'|=1$. 
Taking into account the prefactor $(-1)^n$ in (\ref{eq:Gamma_1D_chiral}), we can conclude that 
       \[ \nu^{\prime\prime}= \begin{cases} 1 & \text{: $(-1)^n\cos k \sin 2\alpha>0$, } \\ -1 & \text{: $(-1)^n\cos k \sin 2\alpha<0$}. \end{cases} \]
\end{enumerate}

Finally, we  use the formula to calculate the topological numbers for
the zero and $\pi$ energy states, $\nu_0$ and $\nu_\pi$, respectively, for
one-dimensional quantum walks with chiral symmetry\cite{AO}:
\[
(\nu_0, \nu_\pi) = \left(
\frac{\nu^\prime+\nu^{\prime\prime}-1}{2}, 
\frac{\nu^\prime-\nu^{\prime\prime}-1}{2} 
\right).
\]
Putting $\nu^\prime$ and $\nu^{\prime\prime}$ into the above formula, we
have
\begin{align}\label{eq:toponum1}
 (\nu_0,\nu_\pi)=
\begin{cases}
(0,-1) & \text{: $(-1)^n\cos k \sin 2\alpha>0$,}\\
(-1,0) & \text{: $(-1)^n\cos k \sin 2\alpha<0$,}\\
\text{undefined} & \text{: $\cos k \sin 2\alpha=0$.}
\end{cases}
\end{align}
The above result is summarized in figures
\ref{fig:topological_numbers}(a) and (b). We remark that topological
numbers $\nu_0$ and $\nu_\pi$ depend  not only upon the angle of the coin
operators $\alpha$ and $n$ but also the wavenumber $k$.
Since the absolute values of both topological numbers is one, we predict
single edge states at quasi-energies of $0$ and $\pi$ for each $k$ using the
bulk-edge correspondence.
Therefore, in the semi-infinite space $V$ where $k$ is continuous in
$[0,2\pi)$, the number of edge states at quasi-energies $0$ and $\pi$ 
become infinite as a consequence of the accumulation of topological
numbers for each $k$. Then, infinitely degenerated edge states develop
wavenumber independent dispersion relations at quasi-energies $0$ and $\pi$,
which is consistent with the proof of the case (6) in Corollary~2.

The above procedure can also be applied to case (1) in
Corollary~2,
\begin{equation}
 \beta \in \{-\alpha+n\pi\;|\; n\in\mathbb{Z},  \alpha \in[0,2\pi)\},
\label{eq:class1}
\end{equation}
where there are neither gaps in the bulk spectrum, nor and there are no edge states.
Putting (\ref{eq:class1}) into (\ref{eq:Gamma_2D}) and
regarding $k_y$ as a fixed parameter $k\in[0,2\pi)$,
the time-evolution operator in the effective one-dimensional
fitted in symmetry time frame is expressed as
\begin{align}
\check{Y} \check{\Gamma}_0^\prime(k_x;k) \check{Y}^{-1} &=
\check{\Gamma}_0^\prime(k_x;k)^{-1},
\nonumber \\
\check{\Gamma}_0^\prime(k_x;k) &=  \hat{D}(k_x/2) H_{\alpha} \hat{D}(k)
 H_{-\alpha}  \hat{D}(k_x/2),
\label{eq:Gamma_2D_chiral_case1}\\
\check{\Gamma}_0^\prime(k_x;k,n) &= (-1)^n \Gamma_0^\prime(k_x;k)\nonumber
\end{align}
with the chiral symmetry operator
\begin{equation}
 \check{Y}= \sigma_2.
\label{eq:Y2}
\end{equation}
More explicitly, $\check{\Gamma}_0^\prime(k_x;k)$ is written as
\begin{align*}
\check{\Gamma}_0^\prime(k_x;k) &= e_0^\prime(k_x) \sigma_0 + i \sum_{j=1,2,3}e_j^\prime(k_x)
 \sigma_j,\\
e_0^\prime(k_x) &= \cos k \cos k_x - \sin k \cos(2\alpha) \sin k_x,\\
e_1^\prime(k_x) &= -\sin k \sin(2\alpha),\\
e_2^\prime(k_x) &= 0,\\
e_3^\prime(k_x) &= -\cos k \sin k_x - \sin k \cos(2 \alpha) \cos k_x.
\end{align*}
We remark that the chiral symmetry operators $\check{Y}$ is
different from $\hat{Y}$ in (\ref{eq:Y1}).
Applying the unitary transform $e^{i\frac{\pi}{4}\sigma_1} \check{\Gamma}_0^\prime(k_x;k) e^{-i
\frac{\pi}{4}\sigma_1}$, the topological numbers for energies of zero and
$\pi$, $\nu_0$ and $\nu_\pi$, are derived through the same procedure as before.
The result is summarized as
\begin{align}\label{eq:toponum2}
 (\nu_0,\nu_\pi)=
\begin{cases}
(-1,0) & \text{: $(-1)^n\sin k \sin 2\alpha>0$,}\\
(0,-1) & \text{: $(-1)^n\sin k \sin 2\alpha<0$,}\\
\text{undefined} & \text{: $\sin k \sin 2\alpha=0$.}
\end{cases}
\end{align}
We summarize the above result in figures
\ref{fig:topological_numbers} (a) and (c).
We note, however, that the topological numbers in (\ref{eq:toponum2}) cannot be
applied to the AE model for the reason explained below.

Remarkably, when we consider case (1) in Corollary~2, which proves
the absence of edge states, topological numbers $\nu_0$ and $\nu_\pi$ become finite,
except $\sin k \sin 2\alpha=0$. 
However, when we predict
the number of edge states using the bulk-edge correspondence for the
system with a boundary, we keep in mind that the boundary does not break
the symmetry retained in the
bulk.
Then, we check whether the boundary realized by the shift operator in 
(\ref{eq:S''}) [or (\ref{S'})] retains chiral symmetry.
Since we consider the shift operator in the position space in
(\ref{eq:S''}), we redefine the chiral symmetry operators as
\[
 \hat{Y}=\sum_{x,y}\ket{x,y}\bra{x,y}\otimes
\begin{bmatrix}
\sigma_1 & 0\\
0 &\sigma_1 
\end{bmatrix},\quad
 \check{Y}=\sum_{x,y}\ket{x,y}\bra{x,y}\otimes
\begin{bmatrix}
\sigma_2 & 0\\
0 &\sigma_2 
\end{bmatrix}.
\]
In order to retain chiral symmetry, the relation  $X S^{\prime\prime} X^{-1}=S^{\prime\prime-1}$
($X \in \{\hat{Y}, \check{Y}\}$) should be satisfied. 
We can confirm that
\begin{align*}
\hat{Y} S^{\prime\prime} \hat{Y} S^{\prime\prime} &=
\sum_{x,y} 
\ket{x,y}\bra{x,y}\otimes
(\ket{R}\bra{R}+\ket{L}\bra{L}+\ket{D}\bra{D}+\ket{U}\bra{U})\\
&=\sum_{x,y} \ket{x,y}\bra{x,y} \otimes I_4
\quad \text{: $(x,y)\in V$,}
\end{align*}
where $I_4$ is an identity matrix of $4\times 4$ matrices, while
\begin{align*}
\check{Y} S^{\prime\prime} \check{Y} S^{\prime\prime} =
\begin{cases}
\sum_{x,y} 
\ket{x,y}\bra{x,y}\otimes I_4
&\text{: $(x,y)\notin \partial V$,}\\
\sum_{x,y} 
\ket{x,y}\bra{x,y}\otimes
(-\ket{R}\bra{R}+\ket{L}\bra{L}+\ket{D}\bra{D}+\ket{U}\bra{U})
&\text{: $(x,y)\in \partial V$.}
\end{cases}
\end{align*}
Therefore, the boundary retains chiral symmetry given by $\hat{Y}$ in
(\ref{eq:Y1}), while it
breaks chiral symmetry defined by the symmetry operator $\check{Y}$ in (\ref{eq:Y2}).
Thus, the topological numbers in (\ref{eq:toponum2}) cannot be
applied to the AE model.

We note that a boundary realized by a shift operator with an extra phase $i$
on the self loop
\begin{align*}
S''_i |x,y\rangle |J\rangle &= S^\prime|x,y\rangle |J\rangle, \; (J\in \{ R,D,U
\}), \\
S''_i |x,y\rangle |L\rangle
 &=
 \begin{cases}
  S'|x,y\rangle |L\rangle & \text{: $(x,y) \notin \partial V$, } \\
  i S'| x,y \rangle | L\rangle & \text{: $(x,y) \in \partial V$,}
 \end{cases}
\end{align*}
satisfies $\check{Y} S_i^{\prime\prime} \check{Y}^{-1} =
S_i^{\prime\prime-1}$, but $\hat{Y} S_i^{\prime\prime} \hat{Y}^{-1} \ne
S_i^{\prime\prime-1}$. With the shift operator $S_i^{\prime\prime}$, we numerically
confirm that a wavenumber independent dispersion relation for the edge states appears and agrees with the
topological numbers in 
(\ref{eq:toponum2}) for the parameters in case (1) in Corollary~2,
whereas there appear to be no edge states in case (6).

\begin{figure}[htbp]
\begin{center}
    \begin{tabular}{c}
     \includegraphics[clip, width=13.20cm]{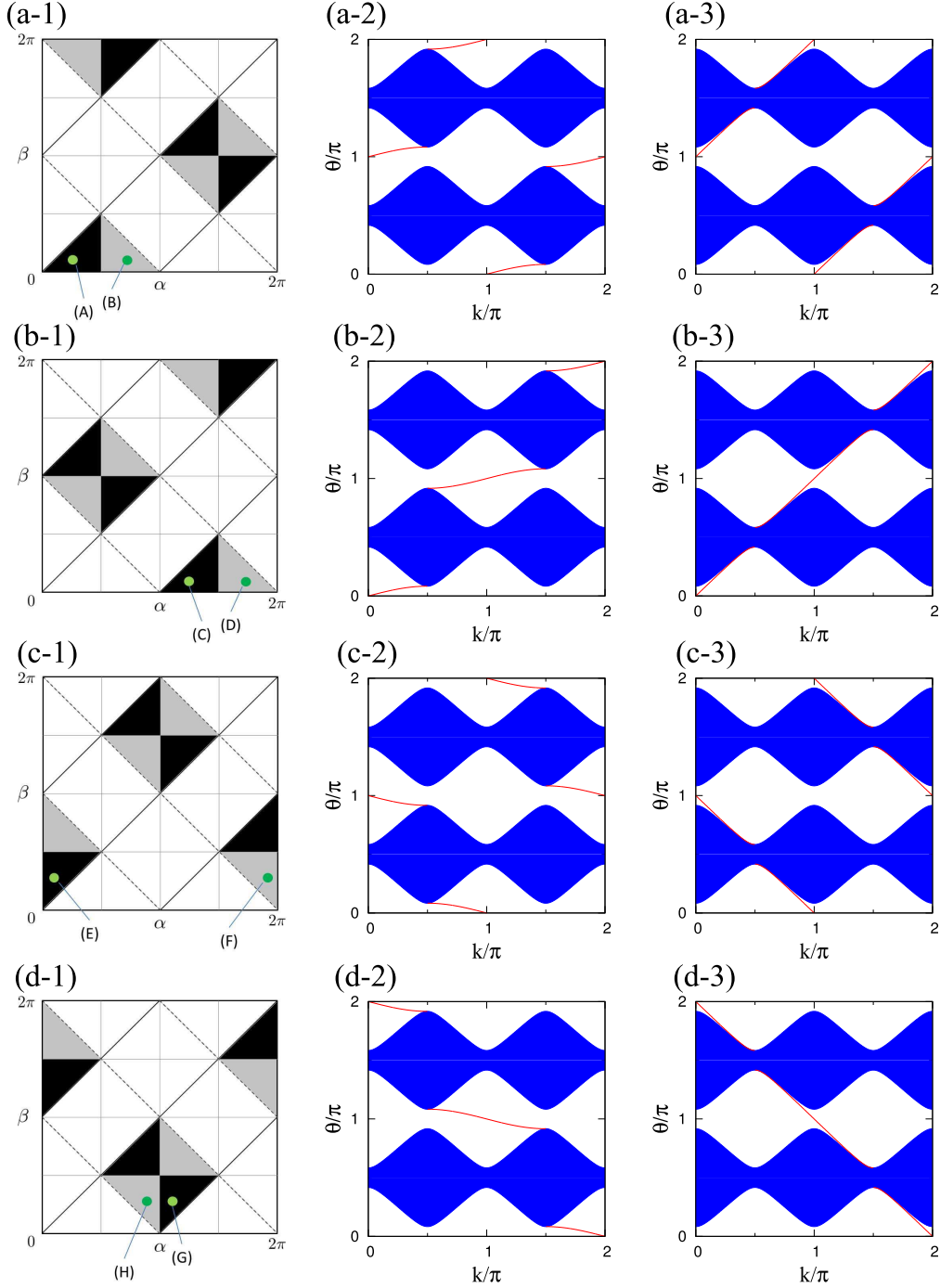}
    \end{tabular}
\end{center}
    \caption{(left column) The black and gray regions represents two
   parameter classes $(\varepsilon_1,\varepsilon_2,\varepsilon_3)$
   considered the middle and right columns, respectively.
              (middle and right columns) The dispersion relations for bulk (filled blue
   regions) and edge states (red curves) for a specific value of
   $(\alpha, \beta)$ included in  two parameter classes shown in the
   left column. The specific values of $(\alpha,\beta)$ are as follows:
   (a-2) (A):$(\pi/4,\pi/6)\in (+,+,+)$, (a-3) (B):$(3\pi/4,\pi/6)\in (-,+,+)$,
   (b-2) (C):$(5\pi/4,\pi/6)\in (+,-,-)$, (b-3) (D):$(7\pi/4,\pi/6)\in
   (-,-,-)$,
   (c-2) (E):$(\pi/6,\pi/4)\in (+,+,-)$, (c-3) (F):$(11\pi/6,\pi/4)\in
   (-,+,-)$,
   (d-2) (G):$(7\pi/6,\pi/4)\in (+,-,+)$, and (d-3) (H):$(5\pi/6,\pi/4)\in
   (-,-,+)$.
   The topological number is $\nu_{2d}=+1$ in (a) and (b), but
   $\nu_{2d}=-1$ in (c) and (d).}
    \label{Fig.2}
\end{figure}
\begin{figure}[htbp]
\begin{center}
    \begin{tabular}{c}
    \includegraphics[clip, width=14.5cm]{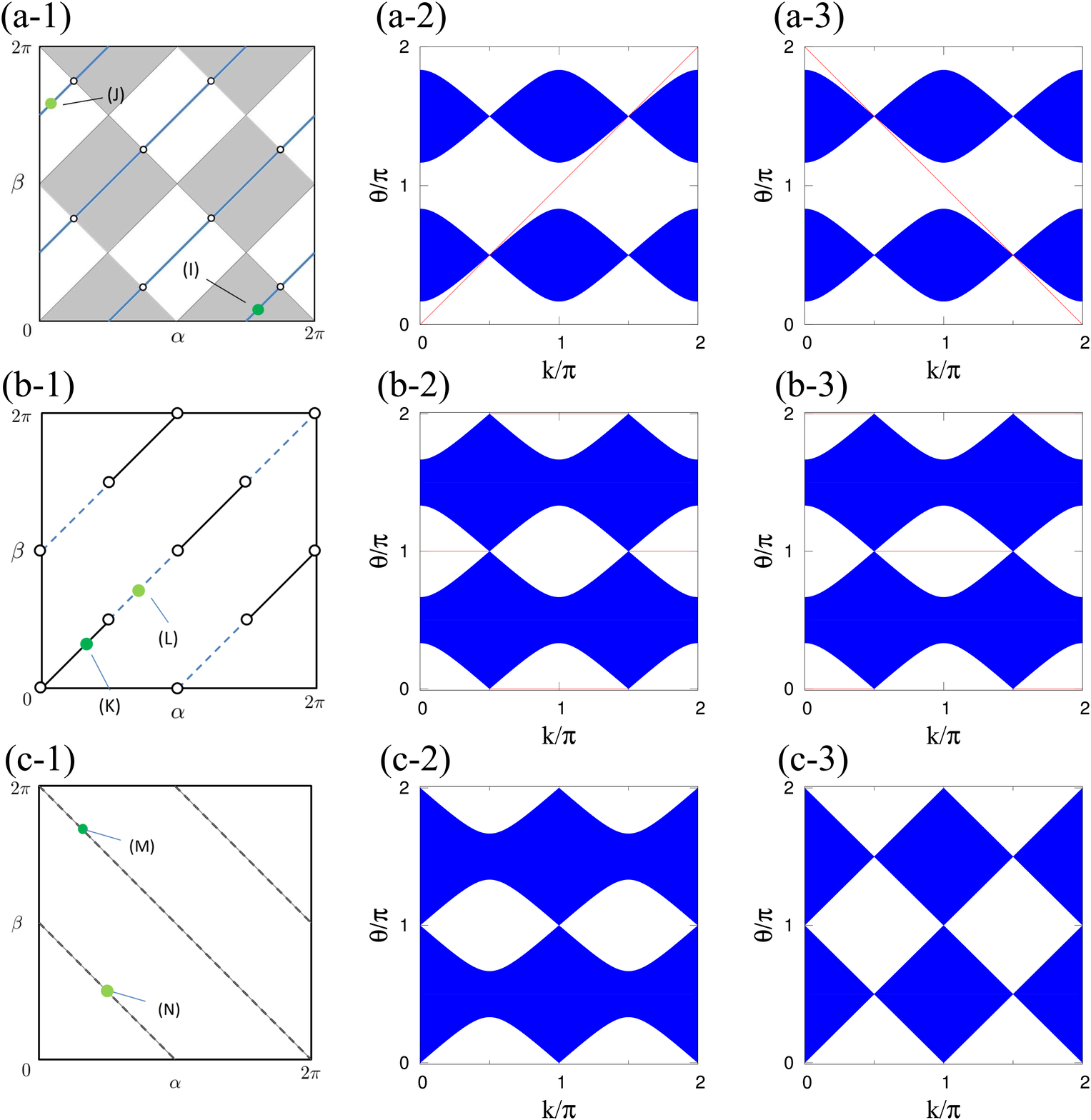}
\end{tabular}\end{center}
    \caption{(left column) The black and gray regions represents two
   parameter classes $(\varepsilon_1,\varepsilon_2,\varepsilon_3)$
   considered the middle and right columns, respectively.
              (middle and right columns) The dispersion relations for bulk (filled blue
   regions) and edge states (red curves) for a specific value of
   $(\alpha, \beta)$ included in  two parameter classes shown in the
   left column. The specific values of $(\alpha,\beta)$ are as follows:
   (a-2) (I):$(5\pi/3,\pi/6)\in (-,-,-)$, (a-3) (J):$(\pi/6,5\pi/3)\in(-,-,+)$,
   (b-2) (K):$(\pi/3,\pi/3)\in (+,+,0)$, (b-3) (L):$(2\pi/3,2\pi/3)\in (+,-,0)$
   (c-2) (M):$(\pi/3,5\pi/3)\in (-,0,+)$ , and (c-3) (N):$(\pi/4,3\pi/4)\in (-,0,-)$
   In (a), the topological number is $\nu_{2d}=+1$. In (b-1), the real lines depict the parameter classes of $(+,+,0)$ 
where the topological numbers are given in figure 2(b). The
   dashed lines depict $(+,-,0)$ where the topological numbers are given
   in figure2(c). In (c), there are no well-defined topological numbers.}
    \label{Fig.7}
\end{figure}
%
\subsection{Demonstrations}

Here, we present plots of dispersion relations for bulk and edge states,
which are written in Theorem~2, 
for specific values of $\alpha$ and $\beta$ in Figures
4,5 and compare them with predictions based on topological numbers.
According to the bulk-edge correspondence, if topological numbers
$\nu_{2d}$, $\nu_0$, and $\nu_\pi$ are finite, we expect
edge states.
We can confirm that the dispersion relations for edge states agree well
with the prediction based on topological numbers.
For the triple of the signs $(\epsilon_1,\epsilon_2,\epsilon_3)$, $(\epsilon_j\in\{\pm\})$, 
where $\epsilon_1=\sgn(\sin 2\alpha\sin2\beta)$,
$\epsilon_2=\sgn(\sin(\alpha+\beta)$, and $\epsilon_3=\sgn(\sin(\alpha-\beta))$,
we have the topological number $\nu_{2d}=\epsilon_2\epsilon_3$. 

\section{Asymptotic behavior of the quantum walk}
We have obtained the dispersion relation of the quantum walk, and
found the existence of the edge spectrum in gaps of the bulk
spectrum, which agrees with the prediction by the bulk-edge
corresponding of topological numbers. Moreover, in the next section,
we shown that the states corresponding to the edge spectrum exponentially localize near self-loops on the boundary. These
observations give a certification to call these states the edge
states. In this section, 
we calculate the finding probability along the boundary.
In the previous studies on the quantum walk on the one dimensional lattice related to the topological phases, 
the response on the finding probability is characterized by a specific property of quantum walks named the localization~\cite{Kitagawa,OK}. 
Since our model treats two-dimensional model, our interest is naturally how the response is deformed in this higher dimensional system. 

To this end, we concentrate upon the finding probability at the boundaries $\partial V$ to observe how the quantum walker recognizes the boundary: 
we set $\nu_n: \mathbb{Z}\to [0,1]$ as the probability that a quantum walker is measured at time $n$ at the boundary,   
	\begin{equation}\label{target} 
        \nu_n(j)=|(U^{2n}\psi_0)(a_j)|^2. 
        \end{equation}
Here $a_j\in A$ is the self-loop at $[0,j]\in V$. 
Recall that the initial state is fixed to the self-loop at the origin $a_0\in A$. 
We have  
	\[ \hat{\psi}_n(k):=\sum_{j\in \mathbb{Z}}(U^{2n}\psi_0)(a_j) \; e^{ikj}=\hat{\varphi}_{n,1}(0;k). \]
Now we define some asymptotic behaviors of $\nu_n$ 
which will be the responses to the cutting edges of the quantum walk as follows. 
	\begin{definition}\label{sto}
        \noindent
        \begin{enumerate}
        \item We say linear spreading occurs at the boundary if and only if there exists 
        a monotonically non-decreasing function $F(\neq 0):\mathbb{R}_{\geq 0}\to [0,1]$, which is not given by $c_0{\sf H}(\cdot)$ 
        with some non-zero constant $c_0$, 
        such that 
        	\[ \lim_{n\to\infty}\sum_{|j|<n y} \nu_n(j)=F(y) \mathrm{\;for\; any\;} y\in \mathbb{R}_{\geq 0}. \]
        Here ${\sf H}(x)$ is the unit step function, that is, 
        	\[ {\sf H}(x)=\begin{cases} 1 & \text{: $x\geq 0$,} \\ 0 & \text{: $x<0$.} \end{cases}\]
        In particular, 
        \begin{enumerate}
        \item if $F(y)$ is continuous on $\mathbb{R}_{\geq 0}$, then we say continuously-linear spreading occurs;
        \item if $F(y)$ is described by $c_0{\sf H}(y-1)$, then we say ballistic spreading occurs. 
        \end{enumerate}
        \item We say localization occurs at the boundary if and only if there exists $j\in\mathbb{Z}$ such that
        	\[ \limsup_{T\to\infty}\frac{1}{T}\sum_{n=0}^{T-1}\nu_n(j)>0. \]
        \end{enumerate}
	\end{definition}
\begin{remark}
If there exists $q\in[0,1)$ and monotonic function $G(\neq 0):\mathbb{R}_{\geq 0}\to [0,1]$ such that 
	\[  \lim_{n\to\infty}\sum_{|j|<n^{q} y} \nu_n(j)=G(y), \]
then 
	\[  \lim_{n\to\infty}\sum_{|j|<n y} \nu_n(j)=\lim_{n\to\infty}G(n^{1-q}y)=\lim_{n\to\infty}c_0{\sf H}(y) \]
holds, where $c_0=\lim_{y\to\infty}G(y)$. 
Thus the condition that $F\neq c_0{\sf H}$ for the definition of the linear spreading prohibits such a slow spreading.  
\end{remark}
We have the following theorem: 
	\begin{theorem}\label{maintheorem}
        Let us define $s:=\sin(\alpha+\beta)\neq 0$, and $r:=\sin(\alpha-\beta)$. 
        \begin{enumerate}
        \item\label{conti_linear} If $0<|r|<1$, then the continuously linear spreading occurs: for any $a<b\in \mathbb{R}$,  
        	\begin{equation}\label{eq:int_g}
                \lim_{n\to\infty}\sum_{a\leq j/n \leq b} \nu_n(j) = \int_a^b g(y)dy
                \end{equation} 
        Here the density of $F(y)$ is 
        	\begin{equation}\label{g} 
                g(y)=\frac{2s^2}{r^2}y^2 f_K(y;|r|) \times \begin{cases} \bs{1}_{[0,\infty)}(y) & \text{: $sr>0$, } \\ \bs{1}_{(-\infty,0]}(y) & \text{: $sr<0$, } \end{cases} 
                \end{equation}
        where $f_K(x;|r|)$ is the Konno function with the parameter $|r|$. 
        \item\label{ballistic} If $|r|= 1$, then the ballistic spreading occurs: 
                \begin{equation}
                \lim_{n\to\infty}\sum_{-\infty \leq j/n \leq y} \nu_n(j) 
                	= s^2 {\sf H}(y-\sgn(rs)) 
                \end{equation}
	which implies the density of $F(y)$ is $s^2\delta(y-\sgn(rs))$. Moreover it holds that 
                \begin{equation}\label{ballistic2}
                \lim_{n\to\infty}\nu_n\left(\sgn(rs)(n-j)\right)=s^2\; \delta_0(j). 
                \end{equation}
        \item\label{localization} If $|r|=0$, then localization occurs: 
        	\begin{equation} \label{localization2}
                \lim_{n\to\infty}\nu_{2n}(j)
                	=
                        \begin{cases}
                        \frac{s^2}{\pi^2}\frac{4}{(j^2-1)^2} & \text{: $j$ is even,} \\
                        0 & \text{: otherwise. }
                        \end{cases}
                \end{equation} 
                \begin{equation}\label{eq:odd}
                \lim_{n\to\infty}\nu_{2n+1}(j)
                	=
                        \begin{cases}
                        \frac{s^2}{4} & \text{: $j\in \{\pm 1\}$,} \\
                        0 & \text{: otherwise. }
                        \end{cases}
                \end{equation}
        \end{enumerate}
        Here $f_K$ is the density of the Konno distribution~\cite{Konno,Konno2} with the parameter $p\in(0,1)$ denoted by 
	\begin{equation}
	f_K(x;p)=\frac{\sqrt{1-p^2}}{\pi (1-x^2)\sqrt{p^2-x^2}} \bs{1}_{(-p,p)}(x). 
	\end{equation}
        \end{theorem}
The dispersion relation classes of the limit theorem in the parameter setting of (\ref{localization}) in Theorem~\ref{maintheorem}, 
that is, $|r|=|\sin(\alpha-\beta)|=0$, 
is $(\epsilon_1,\epsilon_2,0)$, ($\epsilon_j\in \{\pm\}$ ($j=1,2$)). 
In this case, we have a wave number independent dispersion relation, that is, $\theta_0(k)/dk=0$ for all $k$. 
On the other hand, the dispersion relation classes of the limit theorem in the parameter setting of (\ref{ballistic}), 
that is, $|r|=|\sin(\alpha-\beta)|=1$, is a special class of $(-, \epsilon_2, \epsilon_3)$, ($\epsilon_j\in \{\pm\}$ ($j=2,3$)). 
In this case, we have the linear dispersion relation for the edge state, that is, $|\theta_0(k)/dk|=1$ for all $k$. 
The the wave number independence and linearity of the edge state are reflected in the stochastic behaviors of the quantum walk on the boundary as localization and ballistic spreading, respectively. 

\begin{proof}
The $n$-th iteration of the CMV matrix is expressed by a ``quantum" version of the formula corresponding to the Karlin MacGregor formula for random walks: 
	\begin{equation}\label{KM} 
        (\mathcal{C}^n)_{i,j}=\int_{|z|=1} z^n \overline{\chi_i(z)}\chi_j(z) d\mu(z). 
        \end{equation}
Concerning the one-to-one correspondence between the canonical bases of
 $\mathcal{C}$ and the Type-I quantum walk described by Lemma~\ref{CMV}, 
	$\mathcal{C}^n_{0,0}$ expresses the complex-valued amplitude at return to the origin at time $n$, that is, 
        \begin{equation}\label{returnamplitude}
        \mathcal{C}^n_{0,0}=\hat{\psi}_n(k). 
        \end{equation}
Combining (\ref{KM}) with (\ref{returnamplitude}) and Lemma~\ref{spectrumCMV}, we have 
	\begin{equation}
        \hat{\psi}_n(k)=\int_{0}^{2\pi}  e^{in\theta} \left(\frac{w(\theta)}{2\pi}+m_0(k)\delta(\theta_0(k)-\theta) \right) d\theta.
        \end{equation}
The Riemann and Lebesgue lemma implies
	\begin{equation}\label{assymptotic}
        \hat{\psi}_n(k)=e^{in\theta_0(k)} m_0(k)+o(1/\sqrt{n}). 
        \end{equation}
The assumption that $s\neq 0$ ensures $m_0(k)\neq 0$ for almost every $k$. 
We define the Fourier transform of $\nu_n$ by $\phi_n(\xi):=\sum_{j\in \mathbb{Z}} \nu_n(j)e^{i\xi j}$ $(\xi \in \mathbb{R})$. 
Here explicit expression of $\theta_0(k)$ and $m_0(k)$ are described by (\ref{mass}) and (\ref{theta_0}), respectively. 
Now we will prove our statement for each case. \\
{\bf Proof of part (\ref{conti_linear}).}
It holds that
	\begin{equation}\label{integration0}
        \phi_n(\xi)=\int_{0}^{2\pi} \overline{\hat{\psi}_n(k)}\;\hat{\psi}_n(k+\xi) \frac{dk}{2\pi}. 
        \end{equation} 
Under the assumption $\sin(\alpha-\beta)\sin(\alpha+\beta)\neq 0$, then using the expression (\ref{assymptotic}), we obtain
	\begin{equation}\label{integration}
        \phi_n(\xi/n)=\int_{0}^{2\pi} e^{i\xi v(k)}m_0^2(k)\frac{dk}{2\pi} + o(1/\sqrt{n}), 
        \end{equation}
where $v(k)=d\theta_0(k)/dk$:
	\begin{align}\label{v(k)expression} 
        \frac{d\theta_0}{dk} = \mathrm{sgn}(\sin(\alpha-\beta)\sin(\alpha+\beta))|\sin(\alpha-\beta)|\frac{|\cos k|}{\sqrt{1-\sin^2(\alpha-\beta)}\sin^2 k}. 
        \end{align}
Since 
	\[ \frac{d}{dk}\left(\frac{\cos k}{\sqrt{1-r^2\sin^2 k}}\right)
        	= -\frac{\sin k \cos^2(\alpha-\beta)}{(1-\sin^2(\alpha-\beta)\sin^2k)^{3/2}}, \]
we have $0\leq v(k)\leq |r|$ $(rs>0)$ and $-|r|\leq v(k)\leq 0$ $(rs<0)$. 
We consider $rs>0$ case. 
For $0\leq y \leq |r|$, due to the assumption $|r|<1$, the solutions for $y=v(k)$ with respect to $k$ are $k_1<k_2<k_3<k_4$, which satisfy 
	\begin{equation}\label{sin k}
        \sin^2k_j=\frac{r^2-y^2}{r^2(1-y^2)}\;\;(j=1,2,3,4)
        \end{equation}
which implies $0<k_1<\pi/2<k_2<\pi<k_3<3\pi/2<2\pi$, and $k(y):=k_1$, $k_2=\pi-k(y)$, $k_3=\pi+k(y)$ and $k_4=2\pi-k(y)$. 
We divide the integration (\ref{integration}) into the four parts: $\int_0^{\pi/2}+\int_{\pi/2}^{\pi}+\int_{\pi}^{3\pi/2}+\int_{3\pi/2}^{2\pi}$. 
By taking $x=v(k)$ together with (\ref{sin k}), we have 
	\begin{equation}\label{konno1} 
        \frac{dv(k)}{dk}=-\frac{(1-y^2)\sqrt{r^2-y^2}}{\sqrt{1-r^2}}=\frac{1}{\pi f_K(y;r)}. 
        \end{equation}
and 
	\begin{equation}\label{massfanc} 
        m_0^2(k)=\frac{s^2}{r^2}y^2. 
        \end{equation}
Then it holds 
        \begin{align} 
        I_1(\xi) &:= \int_{0}^{\pi/2} e^{i\xi v(k)} m_0^2(k) \frac{dk}{2\pi} \\
        	 &= \int_{0}^{|r|} e^{i\xi y} \frac{s^2}{r^2}y^2 \frac{1}{2}f_K(y;r) dy. 
        \end{align}
In a similar fashion, we have $I_1(\xi)=I_2(\xi)=I_3(\xi)=I_4(\xi)$. 
Therefore we can rewrite (\ref{integration}) by 
	\begin{equation}
        \phi(\xi):=\lim_{n\to\infty}\phi(\xi/n)=\int_{0}^{|r|} e^{i\xi y} \frac{s^2}{r^2}y^2f_K(y;r) dy. 
	\end{equation}
In the same way, for $rs<0$ case we have 
	\begin{equation}
        \lim_{n\to\infty}\phi(\xi/n)=\int_{-|r|}^{0} e^{i\xi y} \frac{s^2}{r^2}y^2f_K(y;r) dy. 
	\end{equation}
Since $\phi(\xi)$ is continuous at $\xi=0$, then the continuity theorem implies that, for any $a<b\in \mathbb{R}$, 
	\[ \lim_{n\to\infty} \sum_{a\leq j/n \leq b} \nu_n(j)=\int_a^b g(y)dy.  \;\;\;\square\] 
\\{\bf Proof of part (\ref{ballistic}).}
Since $|r|=1$, then $m_0(k)$ and $\theta_0(k)$ are reduced to $m_0(k)=|s|$, and 
	\[ \theta_0(k)=
        \begin{cases}  
        \pi+k & \text{: $(\sgn(s),\sgn(r))=(+,+)$,} \\ 
        \pi-k & \text{: $(\sgn(s),\sgn(r))=(+,-)$,} \\
        2\pi-k & \text{: $(\sgn(s),\sgn(r))=(+,-)$,} \\
        k & \text{: $(\sgn(s),\sgn(r))=(-,-)$,} 
        \end{cases}
        \]
respectively. 
From (\ref{assymptotic}) and (\ref{integration0}), the Fourier transform of $\nu_n$ is expressed by 
	\[ \phi_n(\xi/n)\sim s^2\begin{cases} e^{i\xi} & \text{: $(\sgn(s),\sgn(r))\in \{(+,+),(-,-)\}$} \\ 
        e^{i\xi} & \text{: $(\sgn(s),\sgn(r))\in \{(+,-),(-,+)\}$}.  \end{cases} \]
Therefore $F(y)=s^2\bs{1}_{[1,\infty)}(y)$ holds. 
We set $\tau(j):=(U^{2n}\psi_0)(a_j)$. 
Taking the inverse Fourier transform with insertion of these values into $m_0(k)$ and $\theta_0(k)$ in (\ref{assymptotic}), we have 
\begin{align} 
        \tau(n-j) &:= (U^{2n}\psi)(a_{n-j})\sim \int_{0}^{2\pi}e^{in\theta_0(k)}m_0(k) e^{-i(n-j)k} \frac{dk}{2\pi} \\
        	&= (-1)^n|s|\delta_0(j)
        \end{align}
for the $(\sgn(s),\sgn(r))=(+,+)$ and $(-,-)$ cases. 
In the same way, we have the $\tau(-n+j) \sim (-1)^n|s|\delta_0(j)$ $(\sgn(s),\sgn(r))=(+,-)$ and $(-,+)$ cases. $\square$
\\{\bf Proof of part (\ref{localization}).}
Let $B_1=(-\pi/2,\pi/2)\;\mathrm{mod}(2\pi)$, and $B_2=(\pi/2,3\pi/2)$. Since $r=0$, $\theta_0(k)$ is flat, that is, by (\ref{theta_0}), 
	\[ \theta_0(k)= \begin{cases} \pi & \text{: $k\in B_1$} \\ 0 & \text{: $k\in B_2$} \end{cases} \]
for $s>0$, and 
	\[ \theta_0(k)= \begin{cases} 0 & \text{: $k\in B_1$} \\ \pi & \text{: $k\in B_2$} \end{cases}\]
for $s<0$. Moreover, $m_0(k)$ is reduced to $|s||\cos k|$. 
Then we have 
	\[ \hat{\psi}_n(k)=|s|\times \begin{cases} (-1)^n|\cos k| & \text{: $k\in B_1$} \\ |\cos k| & \text{: $k\in B_2$} \end{cases} \;\; (s>0)\]
        \[ \hat{\psi}_n(k)=|s|\times \begin{cases} |\cos k| & \text{: $k\in B_1$} \\ (-1)^n|\cos k| & \text{: $k\in B_2$} \end{cases} \;\; (s<0)\]
For $s>0$, taking the inverse Fourier transform to $\hat{\psi}_n(k)$, we have 
	\begin{align}
        \tau_n(j) &\sim \int_{0}^{2\pi} \hat{\psi}_n(k) e^{-ikj} \frac{dk}{2\pi} \\
        	&= |s| \left\{(-1)^n\int_{k\in B_1} \cos k\; e^{-ikj} \frac{dk}{2\pi}+\int_{k\in B_2} |\cos k|\; e^{-ikj} \frac{dk}{2\pi}\right\}
        \end{align}
We have 
	\begin{align*} 
        J_1(j) &:= \int_{k\in B_1} \cos k\; e^{-ikj} \frac{dk}{2\pi} 
        	= 
                \begin{cases}
                \frac{(-1)^{j/2}}{\pi} \frac{1}{1-j^2} & \text{: $j$ is even, } \\
                1/4 & \text{: $j\in \{\pm 1\}$, } \\
                0 & \text{: otherwise; } 
                \end{cases} \\
        J_2(j) &:=\int_{k\in B_2} |\cos k|\; e^{-ikj} \frac{dk}{2\pi} 
        	= \begin{cases}
                \frac{(-1)^{j/2}}{\pi} \frac{1}{1-j^2} & \text{: $j$ is even, } \\
                -1/4 & \text{: $j\in \{\pm 1\}$, } \\
                0 & \text{: otherwise. } 
                \end{cases}
        \end{align*}
Then it holds that 
	\[ \tau_n(j)\sim |s|((-1)^nJ_1(j)+J_2(j))=
        	\begin{cases}
                \frac{1+(-1)^n}{2} \;\; 2|s|(-1)^{j/2}/(\pi(1-j^2)) & \text{: $j$ is even, } \\
                \frac{1+(-1)^{n-1}}{2}\;\; |s|/2 & \text{: $j\in \{\pm 1\}$, } \\
                0 & \text{: otherwise. } 
                \end{cases}  \]
In the same way for $s<0$, we obtain the same expression for $\tau_n(j)$. Taking the square modulus to $\tau_n(j)$, we obtain the desired conclusion. 
\end{proof}

The group-velocity of the edge state is $v(k)=d\theta_0(k)/dk$, and the effective mass of the edge state is $M(k)=|1/(d^2\theta_0^2(k)/dk^2)|$. 
The limit distribution is expressed by the above physical quantities: 
\begin{corollary}\label{parameterexpression}
Let $v(k)$, $M(k)$ be the above. 
The density function of the limit distribution for $0<|\sin(\alpha-\beta)|<1$ is denoted by $g(\cdot)$. 
Then we have the following parametric plot of the density function:
	\[ \{ (y,g(y)): y\in \mathbb{R} \}=\{ (v(k), 2m_0^2(k)M(k))/\pi: k\in [0,2\pi) \} \]
Here $m_0(k)$ is the density of the edge state at $k\in[0,2\pi)$ defined by (\ref{mass}).
\end{corollary}
Now let's consider the following situation. Assume that we are given the distribution $\nu_n$ ($n>>1$) in advance, which has been
obtained by the measurement of the quantum state $\psi_n$; which is the $n$-th iteration of AE model with the initial state $\psi_0$. 
Here the parameters that determine the AE model; $r=\sin(\alpha-\beta)$ and $s=\sin(\alpha+\beta)$, are not known by us. 
In the rest of this section, we will try to estimate the group velocity $v(k)$ using only this $\nu_n$. 
First, in the following Corollary, we present a means of knowing $v(k)$ using the limit distribution of $\nu_n$; $g(y)$. 
Secondly, we demonstrate a means of estimating $\nu_n$ using this Corollary (see Fig~\ref{fig:Corollary4}). 
\begin{corollary}\label{parametricexpressionCor}
Let $g(\cdot)$, $r$ be the above and $C_0$ be $\int_{-\infty}^\infty g(y)dy$. 
Then the group velocity $v(k)$ satisfies the following: 
	\begin{align} 
        g(v(k))\frac{d v(k)}{dk} &= \frac{2s^2/r^2}{\pi}\;v^2(k), \label{diff}\\   
        \int_{0}^{2\pi} |v(k)| \frac{dk}{2\pi} &= \frac{2\arcsin |r|}{\pi}
        \label{boundary_condition}
        \end{align}
Here $s^2$ is expressed as
	\begin{equation}\label{s^2} s^2=\frac{r^2}{1-\sqrt{1-r^2}}C_0. \end{equation}
\end{corollary}
The second equation is a kind of boundary condition of the above differential equation. 
The above Corollary~\ref{parametricexpressionCor} suggests that we can estimate the underlying spectral information 
of the quantum system from information obtained by observation. 
\begin{proof}
Since the second moment of $f_K(x;r)$ is $1-\sqrt{1-r^2}$~\cite{Konno} and $f_K(x;|r|)$ is an even function, 
then the total mass $C_0$ can be computed explicitly as (\ref{s^2}). (\ref{diff}) comes directly from 
Corollary~\ref{parameterexpression} and (\ref{massfanc}). Concerning the expression of $v(k)$ in (\ref{v(k)expression}), 
we divide the integration into four parts $\int_0^{\pi/2}+\int_{\pi/2}^{\pi}+\int_{\pi}^{3\pi/2}+\int_{3\pi/2}^{2\pi}$ such that 
$v(k)$ is a monotonic function in each region. 
By (\ref{konno1}), 
	\[ \int_{m \pi/2}^{(m+1)\;\pi/2} |v(k)|\frac{dk}{2\pi}
        	=\frac{1}{2}\int_{0}^{|r|} y f_K(y;r)dy,\;\;(m=0,1,2,3). \]
The RHS can be expressed as $\arcsin|r|/(2\pi)$. 
Thus we have (\ref{boundary_condition}). 
\end{proof}
Regarding the explicit expression for $g(x)$ in (\ref{g}), we also notice from this linear differential equation for the group-velocity $v(k)=d\theta_0/dk$ that 
the group velocity of the edge state is expressed by an inverse of Konno's distribution $F_K(y;r)=\int^y f_K(x;r)dx$, that is, 
	\[ v(k)=F_K^{-1}(2s^2 k/(r^2 \pi) + \zeta). \]
Here $\zeta$ satisfies
	\[ \int_{0}^{2\pi} \left|F_K^{-1}(2s^2 k/(r^2\pi)+ \zeta)\right| \frac{dk}{2\pi}= \frac{2\arcsin |r|}{\pi}. \] 

Before closing this section, we compare numerical results for the
probability of the self-loop at $[0,j]$ at the time step $n$, $\nu_n(j)$, 
with Theorem~3. In Figures \ref{fig:linear-spreading} (a) and (b), we 
consider $0<|r|<1$ which gives continuously linear spreading behaviors. 
In Figures \ref{fig:linear-spreading} (a-1) and  (b-1), we see 
that the scaled limit measure $g(y)/n$ in (\ref{g}) runs through the ``middle" of the oscillating $\nu_n(j)$. 
When we take the cumulation of $\nu_n(j)$ to $ny$, that is, $\sum_{j<ny}\nu_n(j)$, 
then in the right figures, we can see that the oscillation is almost
wiped, and the cumulations of $\nu_n(j)$ and the integral of $g(y)$ almost overlap. 
This is a fundamental effect of the convergence in the distribution of Theorem~3.  

Next, we consider the ballistic spreading with
$|r|=1$ and localization with $|r|=0$ in Figure 
\ref{fig:linear-spreading} (c-1). In this case, Theorem~3 and the numerical
results are almost consistent each other. We also confirm that the
numerical result with $|r|=0$ at odd-time steps completely agrees with (\ref{eq:odd}).

Finally, applying Corollary~4, we derive the group velocity $v(k)$ from the numerical result
$\nu_{400}(j)$ in Figure \ref{fig:linear-spreading} as an example. To this end, we need to extract
$g(v(k))$, $r$, and $C_0$ from $\nu_n(j)$. The parameter $r$ is roughly
estimated from Figure \ref{fig:linear-spreading} and
$C_0$ is calculated from $C_0=\sum_{j=0}^{400}\nu_{400}(j)$.
To determine $g(v(k))$, we use the relation $y=j/n=v(k)$ and the fact that
the cumulation of $\nu_n(y)$ shows fewer oscillations, making the numerical fittings stable. 
Assuming that the integral of $g(y)$ is well approximated by a polynomial whose lowest order is third, we have 
\begin{equation}
G(y)=\int_0^y g(y^\prime)
dy^\prime \approx 
\begin{cases} 
\sum_{m=3}^M g_m y^m & \text{: $0\leq y\leq r$,} \\
C_0 & \text{: $y>r$.}
\end{cases}
\label{eq:G(y)}
\end{equation}
We determine the coefficient $g_m$ up to the order $M=5$ by numerically fitting to the
cumulation of $\nu_{n}(j)$ as shown in Figure \ref{fig:Corollary4}(a).
We summarize the extracted values as follows:
\[
 g_3=52.47 \pm 5.64,\  g_4=-483.16\pm54.24,\ g_5=1449.38\pm129.12,\
 r=0.26,\  C_0=0.47456. 
\]
From (\ref{diff}) and (\ref{eq:G(y)}), we obtain
\begin{equation}
 \sum_{m=3}^{M} \frac{m}{m-2}g_m y^{m-2} = \frac{2 s^2/r^2}{\pi} (k-k_0),
\end{equation}
where $k_0$ is a constant of integration. In the case of $M=5$, one of
the general solutions of $y=v(k)$ is guaranteed to be real. Then, we determine $k_0$
such that it satisfies (\ref{boundary_condition}). In figure
\ref{fig:Corollary4}(b), we compare the $|v(k)|$
obtained from the probability $n_{400}(j)$ with the exact solution in 
(\ref{v(k)expression}) and confirm that we can derive $|v(k)|$ from $\nu_n(j)$.

\begin{figure*}[thbp]
\begin{center}
	\includegraphics[width=140mm]{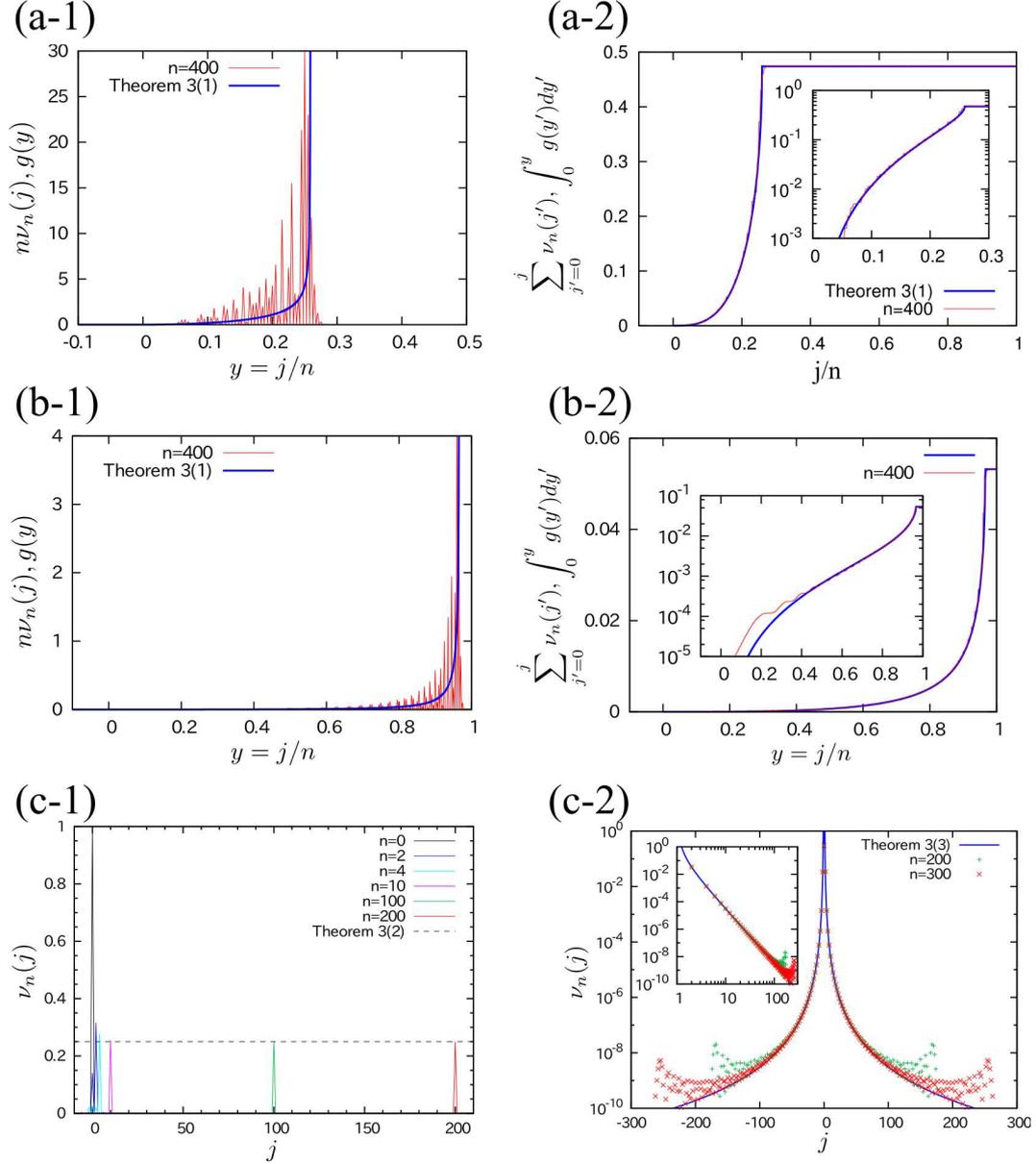}
\end{center}
\caption{ The rescaled measure on self-loops $n \nu_n(j)$ at $n=400$ time steps, as 
 obtained by numerical simulations (a-1) and (b-1), and the cumulative distributions (thin red curves)
 are obtained by (a-2) and (b-2). 
 Figures (a) and (b) are results for the case of $(\alpha,\beta)=(\pi/4, \pi/6)$ whose 
 dispersion relations are given in (a-2) in Figure 4 and 
 the case of $(\alpha,\beta)=(3\pi/4, \pi/6)$, whose 
 dispersion relations are given in (a-3) in Figure 4, respectively.
 The corresponding limit density, $g(y)$, in (\ref{g}) and the distribution in (\ref{eq:int_g}) in
 Theorem~\ref{maintheorem} (1) are shown by thick blue curves. Inset of
 (a-2 and (b-2)): the semi-$log$ plot of the distribution.
(c): The measures on self-loops $n \nu_n(j)$ at various time steps, 
 as obtained by the numerical simulations with (c-1)
 $(\alpha,\beta)=(5\pi/3, \pi/6)$ and (c-2) $(\pi/3,\pi/3)$ 
 with dispersion relations given in (a-2) and (b-2) in
Figure 5, respectively. (c-1) The corresponding limit measures $\lim_{n\rightarrow
 \infty} \nu_n(j)$ in (\ref{ballistic2}) in
 Theorem~\ref{maintheorem} (2) is shown by the dashed line.
(c-2) The corresponding limit measures $\lim_{n\rightarrow
 \infty} \nu_n(j)$ in (\ref{localization2}) in
 Theorem~\ref{maintheorem} (3) is shown by the thick blue curves. Inset
 of (c-2):
 the log-log plot of the main figure.
}
\label{fig:linear-spreading}
\end{figure*}

\begin{figure*}[htbp]
\begin{center}
	\includegraphics[width=160mm]{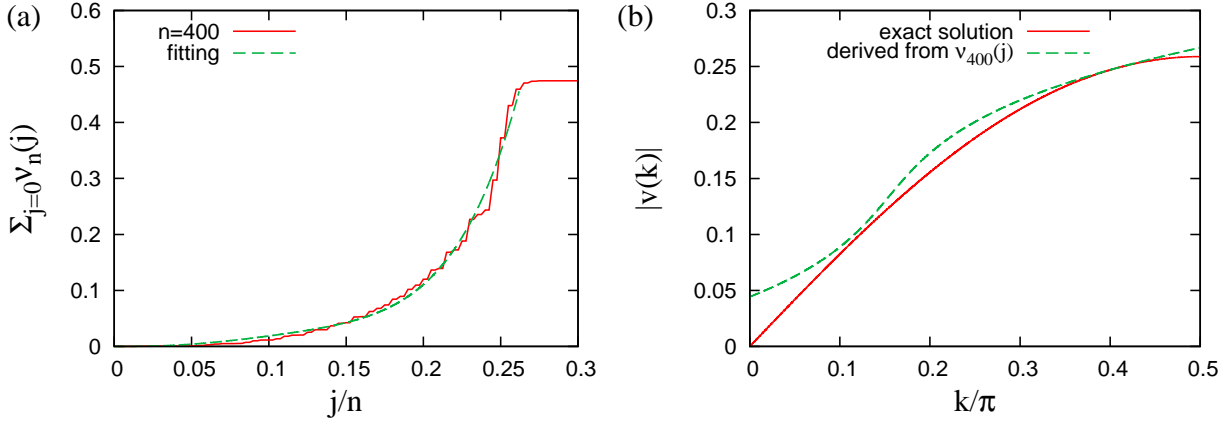}
\end{center}
\caption{
(a) The cumulation of $\nu_{400}(j)$ (solid red curve) for the $(\alpha,\beta)=(\pi/4,\pi/6)$ parameter case
 and $G(y)=\sum_{m=3}^5 g_m y^m$ obtained by the numerical fitting (dashed green curve).
(b) The group velocity $|v(k)|$ derived from  $\nu_{400}(j)$ with $k_0=-0.602$(green
 dashed curve) and the
 exact solution in (\ref{v(k)expression}) (solid red curve).
\label{fig:Corollary4}
}
\end{figure*}

\section{Limit distributions toward bulk}
In the previous section, we showed limit theorems for the finding probability at the self-loops on the boundary. 
In this section, 
we will show how the contribution of the edge states to the behavior of the quantum walk decays toward the bulk. 
To this end, we provide a limit theorem corresponding to Theorem~3 for the finding probability at other arcs 
for $0<|r|<1$ and $s\neq 0$ case. 

We set $\varphi'_n:=\hat{\Gamma}^n_k \varphi'_0$ with $\varphi'_0(j)=\delta(j){}^T[0,1]$, 
where $\hat{\Gamma}: \ell^2(\mathbb{Z}_+;\mathbb{C}^2)\to \ell^2(\mathbb{Z}_+;\mathbb{C}^2)$ is defined in (\ref{Gammahat}). 
Recall that we have shown that $\varphi'_n$, which is the $n$-th iteration of the quantum walk on the Fourier space, 
is expressed by the $n$-th power of the CMV matrix $\mathcal{C}_k$ in (\ref{unitaryequiv2}), that is, 
	\begin{equation*}
        \varphi'_n= \Lambda_k^{-1} ({}^T\mathcal{C}_k)^{n} \Lambda_k\varphi'_0. 
        \end{equation*} 
We have 
	\begin{align}
        \varphi'_n(j) &= \begin{bmatrix}  e^{-iw(2j+1)}(\mathcal{C}_k)_{0,2j+1} \\ e^{-iw(2j)}(\mathcal{C}_k)_{0,2j} \end{bmatrix} \notag \\
        	&= \int_{|z|=1} z^n \begin{bmatrix}  e^{-iw(2j+1)}x_{2j+1}(z) \\ e^{-iw(2j)}x_{2j}(z) \end{bmatrix} d\mu(z), 
        \end{align}
where $x_{j}(z)=\overline{\chi_j(1/\bar{z})}$. 
The first equality derives from the definition of $\Lambda_k$ with $(\Lambda_k\varphi'_0)(j)=\delta(j)$ and 
$({}^T\mathcal{C}_k)_{j,i}=(\mathcal{C}_k)_{i,j}$; the second equality is given by (\ref{KM}). 
This Lemma follows from \cite{CGMV,KS}: 
\begin{lemma}\label{lemKS}
Let $\mathcal{C}$ be the CMV matrix with the Verblunsky parameter $(\gamma,0,\gamma,0,\dots)$. 
Assume $\mathrm{Re}(\gamma)\neq 0$ which is a necessary and sufficient condition for $\sigma_p(\mathcal{C})\neq \emptyset$. 
Then $\sigma_p(\mathcal{C})=\{e^{i\beta}\}$ satisfies (\ref{eigenvalue}) and 
	\begin{align*}
        x_{2j}(e^{i\beta}) = \lambda^j,\; x_{2j+1}(e^{i\beta})=\lambda^{2j+1}, 
        \end{align*} 
where 
	\[ \lambda = \frac{\sgn(\mathrm{Re}(\gamma))}{\rho}\left( \sqrt{1-\mathrm{Im}^2(\gamma)}-|\mathrm{Re}(\gamma)| \right). \]
\end{lemma}
Lemma~\ref{lemKS} and the Riemann-Lebesgue lemma imply 
	\begin{equation}\label{assymptotics1} 
        \varphi_n'(j) = e^{in \theta_0(k)} m_0(k) \begin{bmatrix}  e^{-iw(2j+1)}\lambda^{2j+1}(k) \\ e^{-iw(2j)}\lambda^{2j}(k)\end{bmatrix} +o(1/\sqrt{n}),   
        \end{equation}
where 
	\[ \lambda(k)=\frac{\sgn(\mathrm{Re}(\eta(k)))}{\rho(k)}\left( \sqrt{1-\mathrm{Im}^2(\eta(k))}-|\mathrm{Re}(\eta(k))| \right). \]
Recall that $U:\mathcal{A}\to \mathcal{A}$ is the unitary time evolution of the quantum walk and $\psi_0\in \mathcal{A}$ is the initial state. 
Also recall that the relabelling of each element of $A$ in Definition~1. 
We set $\psi_n=U^{2n}\psi_0$. For $(j,m)\in V$ and $d\in \{0,1\}$, 
the square modulus of $\psi_n((j,m);d)$ 
is the probability that we find the state $((j,m);d)$ after the $n$-th iteration of $U$ starting from $((0,0);1)$. 
Thus we newly introduce the probability measure on $\nu_n: \mathbb{Z}_{+}\times \mathbb{Z} \to [0,1]$ such that
	\[ \nu_n(2j,m):=|\psi_n((j,m);1)|^2 \mathrm{\;\;and\;\;} \nu_n(2j+1,m):=|\psi_n((j,m);0)|^2. \]
Remark that $\nu_n(0,m)$ is identical to $\nu_n(m)$ as discussed in the previous section. 
We set the Fourier transform of them such that for $j\in \mathbb{Z}_{+}$, $n\in \mathbb{Z}_{+}$ and $\xi\in \mathbb{R}$, 
	\begin{align*}
        \phi_n^{(j)}(\xi) &= \sum_{m\in \mathbb{Z}} \nu_n(j,m) e^{i\xi m}.
        \end{align*}
Remark that $\phi_n(0,\xi)$ coincides with $\phi_n(\xi)$ in the previous section. 

In the above discussion, we fixed $k$ and treated $\hat{\varphi}_n$ as $\varphi'(\cdot)=\hat{\varphi}_n(\;\cdot\;;k)\in \ell^2(\mathbb{Z}_{+};\mathbb{C}^2)$  
whereas below, we will treat $\hat{\varphi}_n$ as a function of $k$ for fixed $j$, 
that is, $\hat{\varphi}_n(j;\;\cdot\;)\in L^2([0,2\pi);\mathbb{C}^2)$. 
We set $\hat{\varphi}_n(j;k) :={}^T[\hat{\varphi}_{n,0}(j;k),\hat{\varphi}_{n,1}(j;k)]$,  
and obtain
	 \begin{align} 
         \phi_n^{(2j)}(\xi) &= \int_0^{2\pi} \overline{\hat{\varphi}_{n,1}(j;k)}\hat{\varphi}_{n,1}(j;k+\xi) \frac{dk}{2\pi}, \\
         \phi_n^{(2j+1)}(\xi) &= \int_0^{2\pi} \overline{\hat{\varphi}_{n,0}(j;k)}\hat{\varphi}_{n,0}(j;k+\xi) \frac{dk}{2\pi}.  
         \end{align}
Since $r\neq 0$, by replacing $\xi$ with $\xi/n$ in the above equation and using expansion (\ref{assymptotics1}), we can obtain 
	\begin{equation}
        \begin{bmatrix}\phi^{(2j+1)}_n(\xi/n) \\ \phi^{(2j)}_n(\xi/n) \end{bmatrix} 
        = \int_{0}^{2\pi} e^{i\xi v(k)}m_0^2(k)\begin{bmatrix}  \lambda^{2(j+1)}(k) \\ \lambda^{2j}(k) \end{bmatrix} \frac{dk}{2\pi} + o(1/\sqrt{n}).
        \end{equation}
Since $0<|r|<1$, by putting $y=v(k)$ for $k\in [0,2\pi)$, we have 
	\[ \lambda^2(k)=\frac{|r|-|s||y|}{|r|+|s||y|}. \]
Combining the above with (\ref{konno1}) and (\ref{massfanc}), we obtain
	\[ \lim_{n\to\infty}\phi^{(j)}_n(\xi/n)  
        = \int_{-\infty}^{\infty} e^{i\xi y} g(j,y)  dy. \]
Here the measure $g:\mathbb{Z}_+\times \mathbb{R}\to \mathbb{R}_{+}$ is denoted by 
	\begin{align}
        g(2j+1,y) 
        & = \frac{s^2}{r^2}y^2f_K(y;|r|) \zeta^{j+1}(y) 
        \times \begin{cases} \bs{1}_{[0,\infty)}(y) & \text{: $sr>0$,} \\  \bs{1}_{(-\infty,0]}(y) & \text{: $sr<0$,} \end{cases} \notag \\ 
        \label{measure} \\ 
        g(2j,y)
        & = \frac{s^2}{r^2}y^2f_K(y;|r|)  \zeta^{j}(y) 
        \times \begin{cases} \bs{1}_{[0,\infty)}(y) & \text{: $sr>0$,} \\  \bs{1}_{(-\infty,0]}(y) & \text{: $sr<0$,} \end{cases} \notag 
        \end{align} 
with 
	\begin{equation}\label{zeta} \zeta(y)=\frac{|r|-|s||y|}{|r|+|s||y|}.  \end{equation}
We summarize this section in the following theorem which extends the result of Theorem~\ref{maintheorem} (1) to the other arcs since $g(0,y)$ coincides with $g(y)$ given by (\ref{g}): 
\begin{theorem}
Let $\nu_n: \mathbb{Z}_+\times \mathbb{Z}\to [0,1]$ and $g: \mathbb{Z}_+\times \mathbb{R}\to \mathbb{R}_+$ be the above. 
For $0<|r|<1$ with $s\neq 0$, we have  
	\[ \lim_{n\to \infty} \sum_{a<m/n<b } \nu_n(j,m)= \int_a^b g(j,y) dy \;\; (j\in \mathbb{Z}_+)\]
with $g(j,y)$ given by (\ref{measure}). 
\end{theorem}
Remark that the limit measure $g(j,\cdot)$ has the support $[0,|r|)$ or $(-|r|,0]$. 
If we insert $y=\sgn (rs) r$ into (\ref{zeta}), then $\zeta(y)=(1-|s|)/(1+|s|)$ holds, which agrees with (B21) in \cite{AE}; 
this is the decay rate of the edge state wavefunctions 
towards the bulk for the zero quasi energy. The normalized position $y=\sgn (rs) r$ coincides with the maximal absolute value of the group-velocity $v(k)$. 
In such a wave number $k$, the effective mass $M(k)$, which is the inverse of the derivative of $v(k)$, diverges implying 
that the value $g(j; y)$ also diverges, see Corollary~\ref{parameterexpression}. 

We remark that the LHS in Theorem~4 for the $s=0$ case, which is equivalent to the condition lacking the edge state, 
becomes zero and also note that the limit measure $g: \mathbb{Z}_+\times \mathbb{R}\to \mathbb{R}_+$ 
is not a probability measure; we have focused on only the edge state's effect on the finding probability 
without considering on the bulk state. Indeed, putting $\Omega:=\mathbb{Z}_+\times \mathbb{R}$ and supposing $rs>0$, we have 
	\begin{align*}
        \int_{\omega\in \Omega} g(\omega) d\omega &= \int_{-\infty}^{\infty} \sum_{j=0}^\infty g(j;y) dy = \int_{0}^{|r|} \frac{1+\zeta(y)}{1-\zeta(y)} \frac{s^2}{r^2} y f_K(y;|r|) dy \\
        &= \int_{0}^{|r|} \frac{|s|}{|r|} y f_K(y;|r|) dy \leq |s|/2 <1. 
        \end{align*}
Consideration of the missing value of at least $1-|s|/2$ which is the contribution of the bulk state 
may be an interesting subject for future work.

\par
\
\par\noindent
\noindent
{\bf Acknowledgments.}
\par
NK were supported by 
the Grant-in-Aid for Scientific Research Challenging Exploratory Research (JSPS KAKENHI No.\ JP15K13443).
ES acknowledges financial support from 
the Grant-in-Aid for Young Scientists (B) and of Scientific Research (B) Japan Society for the Promotion of Science (Grant No.~16K17637, No.~16K03939). 
HO was supported by a Grant-in-Aid for Scientific Research on Innovative
Areas ``Topological Materials Science'' (JSPS KAKENHI No.\ JP16H00975)
and also JSPS KAKENHI (No.\ JP16K17760 and No.\ JP16K05466). 
Finally, authors would like to thank reviewers for providing invaluable suggestions to this paper. 
\par

\begin{small}
\bibliographystyle{jplain}

\end{small}

\appendix
\def\thesection{Appendix \Alph{section}}
\renewcommand{\theequation}{A.\arabic{equation}}
\setcounter{equation}{0}


\section{Proof of Lemma~\ref{alternative}}\label{pfalternative}
We set $C'=\mathcal{U}C\mathcal{U}^{-1}$ and $S'=\mathcal{U}S\mathcal{U}^{-1}$. 
By the bipartiteness of the coin operator $C$ with respect to the ``$\{0,1\}$" and ``$\{2,3\}$" directions, it holds that for every $\phi\in \mathcal{A}$, 
	\begin{align}
        \begin{bmatrix} (C\phi)(\bs{x};0) \\ (C\phi)(\bs{x};1)\end{bmatrix}=\sigma_1 H_\alpha \begin{bmatrix}\phi(\bs{x};2)\\ \phi(\bs{x};3) \end{bmatrix},\;\;
        \begin{bmatrix} (C\phi)(\bs{x};2) \\ (C\phi)(\bs{x};3)\end{bmatrix}=\sigma_1 H_\beta \begin{bmatrix}\phi(\bs{x};0)\\ \phi(\bs{x};1) \end{bmatrix}. 
        \end{align}
Here $\sigma_1=|0\rangle\langle 1|+|1\rangle\langle 0|$. 
Using this, for every $\psi\in \ell^2(V;\mathbb{C}^4)$ with $\psi(\bs{x})={}^T[\psi_0(\bs{x}),\dots,\psi_3(\bs{x})]$, we have 
	\begin{equation}\label{C'} 
        (C'\psi)(\bs{x})
        = {|0\rangle}\otimes \sigma_1 H_\alpha \begin{bmatrix}  \psi_2(\bs{x}) \\ \psi_3(\bs{x}) \end{bmatrix} 
        + {|1\rangle}\otimes \sigma_1 H_\beta\begin{bmatrix}  \psi_0(\bs{x}) \\ \psi_1(\bs{x}) \end{bmatrix}. 
        \end{equation}
Concerning the shift operator $S$ flipping the arc direction; $(S\phi)(a)=\phi(\bar{a})$ , we have 
	\begin{equation}\label{S'} 
        (S'\psi)(x,y)
        	=\begin{cases}
                {}^T[\psi_1(x+1,y)\;\;\psi_0(x-1,y)\;\;\psi_3(x,y+1)\;\;\psi_2(x,y-1)] & \text{: $(x,y)\notin \partial V$,} \\
                {}^T[\psi_1(x+1,y)\;\;\psi_1(x,y)\;\;\psi_3(x,y+1)\;\;\psi_2(x,y-1)]
		  & \text{: $(x,y) \in \partial V$.}
                \end{cases}
        \end{equation}
By (\ref{C'}) and (\ref{S'}), we have 
	\begin{multline}\label{before}
        (U'\psi)(x,y)=
                \tilde{P}_\alpha \psi^{(\updownarrow)}(x+1,y)+ \tilde{Q}_\alpha\psi^{(\updownarrow)}(x-1,y) \\
                	+\tilde{P}_\beta \psi^{(\leftrightarrow)}(x,y+1)+ \tilde{Q}_\beta\psi^{(\leftrightarrow)}(x,y-1), \;\;((x,y)\notin \partial V),
        \end{multline}
and
        \begin{multline}
        (U'\psi)(x,y)=
                \tilde{P}_\alpha \psi^{(\updownarrow)}(x+1,y)+ \tilde{S}_\alpha\psi^{(\updownarrow)}(x,y) \\
                	+\tilde{P}_\beta \psi^{(\leftrightarrow)}(x,y+1)+ \tilde{Q}_\beta\psi^{(\leftrightarrow)}(x,y-1), \;\;((x,y)\in \partial V).
        \end{multline}
Then we prove the lemma. $\square$
\section{Proof of Lemma~\ref{alternative2}}\label{pfalternative2}
For $\psi=\psi^{(\leftrightarrow)}\in \ell^2(V;\mathbb{C}^4)$, 
we first examine $({U'}^2\psi)(x,y)$. 
We should remark that 
	\begin{equation}\label{flip}
        (U'\psi^{(\leftrightarrow)})(x,y)=(U'\psi)^{(\updownarrow)}(x,y),
        \end{equation}
because of $\tilde{Q}_\beta$ and $\tilde{P}_\beta$.
Thus by (\ref{odd}) in Lemma~\ref{alternative}, for $(x,y)\notin \partial V$, we have 
	\begin{align}
        ({U'}^2\psi)(x,y) &= \tilde{P}_\alpha (U'\psi)^{(\updownarrow)}(x+1,y)+\tilde{Q}_\alpha (U'\psi)^{(\updownarrow)}(x-1,y). 
        \end{align}
Moreover by (\ref{even}) in Lemma~\ref{alternative} and (\ref{flip}), we have 
	\begin{multline}
        ({U'}^2\psi)(x,y) = \tilde{P}_\alpha (\tilde{Q}_\beta \psi^{(\leftrightarrow)}(x+1,y-1)+\tilde{P}_\beta \psi^{(\leftrightarrow)}(x+1,y+1)) \\ 
        	+\tilde{Q}_\alpha (\tilde{Q}_\beta \psi^{(\leftrightarrow)}(x-1,y-1)+\tilde{P}_\beta \psi^{(\leftrightarrow)}(x-1,y+1)).
        \end{multline}
Since $\psi$ is expressed by some $\varphi\in \mathcal{U}_e(\mathcal{A}^{(\leftrightarrow)})$ such that $\psi(\bs{x})=|0\rangle\otimes \varphi(\bs{x})$ and 
it holds $\tilde{X}_\alpha\tilde{Y}_\beta=|0\rangle\langle 0|\otimes X_\alpha Y_\beta$ $(X,Y\in \{P,Q\})$, we have 
	\begin{multline}
        ({U'}^2\psi)(x,y) = |0\rangle \otimes \{P_\alpha Q_\beta \varphi(x+1,y-1)+P_\alpha P_\beta \varphi(x+1,y+1)) \\ 
        	+Q_\alpha Q_\beta \varphi(x-1,y-1)+ Q_\alpha P_\beta \varphi(x-1,y+1))\}. 
        \end{multline}
In the same way for the $(x,y)\in \partial V$ case, 
	\begin{multline}
        ({U'}^2\psi)(x,y) = |0\rangle \otimes \{P_\alpha Q_\beta \varphi(x+1,y-1)+P_\alpha P_\beta \varphi(x+1,y-1)) \\ 
        	+S_\alpha Q_\beta \varphi(x,y-1)+ S_\alpha P_\beta \varphi(x,y+1))\}, 
        \end{multline}
which implies that
	\begin{equation}\label{Eq1} ({U'}^2\psi)(\bs{x}) = |0\rangle \otimes (\Gamma\varphi)(\bs{x}), \;\;(\forall \bs{x}\in V). \end{equation}

By the way, for any $\phi\in \mathcal{A}^{(\leftrightarrow)}$, we have 
	\[ (U^2|_{\mathcal{A}^{(\leftrightarrow)}}\phi)(\bs{x};j)=(\mathcal{U}^{-1}{U'}^2\mathcal{U}\phi)(\bs{x},j) =\begin{cases} ({U'}^2\psi)_j(\bs{x})=(\Gamma\varphi)_j(\bs{x}) & \text{: $j\in\{0,1\}$} \\ 0 & \text{: $j\in \{2,3\}$}. \end{cases}  \]
where $\psi(\bs{x})=(\mathcal{U}\phi)(\bs{x})=|0\rangle \otimes \varphi(\bs{x})$. 
Here in the last equality, we used (\ref{Eq1}). 
Remarking $\mathcal{U}\mathcal{U}^{-1}_e\varphi=\varphi$, we have 
	\begin{align*} 
        (\mathcal{U}_e U^2|_{\mathcal{A}^{(\leftrightarrow)}} \mathcal{U}_e^{-1}\varphi)(\bs{x})
        	&= \begin{bmatrix} (U^2|_{\mathcal{A}^{(\leftrightarrow)}} \mathcal{U}_e^{-1}\varphi)(\bs{x};0) \\ (U^2|_{\mathcal{A}^{(\leftrightarrow)}} \mathcal{U}_e^{-1}\varphi)(\bs{x};1) \end{bmatrix} 
                =  \begin{bmatrix} (\Gamma\varphi)_0(\bs{x}) \\ (\Gamma\varphi)_1(\bs{x}) \end{bmatrix}
                =(\Gamma\varphi)(\bs{x}),
        \end{align*}
which completes the proof. $\square$
\section{Connection between the moving shift and flip-flop shift representations of the quantum walk}
For clarity, we mention the connection between the above description
of the coin and the shift operators $C'$ and $S'$ in (\ref{C'}) and (\ref{S'}), which is so called the flip-flop shift representation, and the
commonly used description of the coin and the shift operators,
$C^{\prime\prime}$ and $S^{\prime\prime}$, respectively such as in
Ref.\ \cite{AE}, which is so called the moving shift representaion. 
We regard four internal states in the vertex representation as the
left,
right, down, and up moving components,
\[
|L\rangle = {}^T[1,0,0,0],\;\; 
|R\rangle = {}^T[0,1,0,0],\;\; 
|D\rangle = {}^T[0,0,1,0],\;\; 
|U\rangle = {}^T[0,0,0,1],\;\; 
\]
respectively, and the position $(x,y)$ is described by
\[
 \ket{x,y}.
\]
On this basis, the coin operator $C^\prime$ is described by
\begin{align*}
C^{\prime} &=
\sum_{x,y} 
\ket{x,y}\bra{x,y}\otimes
\begin{bmatrix}
0 & \sigma_1 H_\alpha\\
\sigma_1 H_{\beta} & 0
\end{bmatrix}\;\;:(x,y) \in V.
\end{align*}
while the common descriptions of $C^{\prime\prime}$ and $S^{\prime\prime}$ are
\begin{align*}
C^{\prime\prime} &=
\sum_{x,y} 
\ket{x,y}\bra{x,y}\otimes
\begin{bmatrix}
0 & H_\alpha\\
H_{\beta} & 0
\end{bmatrix}\;\;:(x,y) \in V.
\end{align*}
and
\begin{align}
S''|x,y\rangle | R\rangle &= | x+1,y \rangle | R\rangle, \nonumber\\
S''|x,y\rangle | D\rangle &= | x,y-1 \rangle | D\rangle,  \nonumber\\
S''|x,y\rangle | U\rangle &= | x,y+1 \rangle | U\rangle, \label{eq:S''}\\
S''|x,y\rangle | L\rangle
 &=
 \begin{cases}
  | x-1,y \rangle | L\rangle & \text{: $(x,y) \notin \partial V$,} \\
  | x,y \rangle | R\rangle & \text{: $(x,y) \in \partial V$,}
 \end{cases}\nonumber
\end{align}
respectively. 
Rewriting the time-evolution operator $U^\prime$, we have
\begin{align}
 U^\prime = S^\prime C^\prime = S^{\prime\prime} (S^{\prime\prime-1}
 S^\prime C^\prime).
\label{eq:U'rewrite}
\end{align}
Taking into account the different order of internal states in the basis,
$S^{\prime\prime-1} S^\prime$ is expressed as
\begin{align*}
S^{\prime\prime-1} S^\prime &=
\sum_{x,y}
\ket{x,y}\bra{x,y}\otimes
\begin{bmatrix}
\sigma_1 & 0\\
0 & \sigma_1
\end{bmatrix},
\end{align*}
and we confirm the relation
\begin{align}
C^{\prime\prime} = S^{\prime\prime-1}  S^\prime C^\prime.
\label{eq:C''_C'}
\end{align}
From (\ref{eq:U'rewrite}) and (\ref{eq:C''_C'}), without loss of generality, we can alternatively describe the
model as $U^\prime=S^{\prime\prime} C^{\prime\prime}$ 
instead of $U^{\prime}=S^\prime C^\prime$ in the vertex representation.

\end{document}